\documentclass[a4paper,11pt]{article}
\usepackage{jpub} 
\usepackage{amsthm}
\usepackage{shuffle}
\usepackage{mathtools}
\usepackage{tensor}
\usepackage{mdframed}
\usepackage{graphicx}
\usepackage{caption}
\usepackage{subcaption} 
\usepackage{float}
\usepackage{amsfonts}
\usepackage{tikz}
\usepackage{simplewick}
\usepackage{circuitikz}
\usetikzlibrary{ decorations.markings}
\usetikzlibrary{cd}

\usepackage[english]{babel}
\theoremstyle{definition}
\newtheorem{definition}{Definition}[section]
\newtheorem{lemma}{Lemma}[section]
\newtheorem{corollary}[lemma]{Corollary}
\newcommand\Gto{\stackrel{\mathclap{\normalfont\mbox{$G_n'$}}}{\longrightarrow}}

\newcommand{\precdot}{\prec\mathrel{\mkern-5mu}\mathrel{\cdot}}
 
\newcommand{\calC}{{\mathcal{C}}}
\newcommand{\calH}{{\mathcal{H}}}
\newcommand{\calV}{{\mathcal{V}}}
\newcommand{\calP}{{\mathcal{P}}}

\newcommand{\RR}{\mathbb{R}}

\newcommand{\dcf}{decoherence functional}

\title{In-in correlators and scattering amplitudes\\ on a causal set}

\author[a]{Emma Albertini,}
\author[a,b]{Fay Dowker,}
\author[a]{Arad Nasiri}
\author[a]{and Stav Zalel\note{Corresponding author: stav.zalel11@imperial.ac.uk}}

\affiliation[a]{Blackett Laboratory, Imperial College London, SW7 2AZ, U.K}
\affiliation[b]{Perimeter Institute, 31 Caroline Street North, Waterloo ON, N2L 2Y5, Canada}

\abstract{Causal set theory is an approach to quantum gravity in which spacetime is fundamentally discrete at the Planck scale and takes the form  of an irregular Lorentzian lattice, or ``causal set'', from which continuum spacetime emerges in a large-scale (low-energy) approximation. In this work, we present new developments in the framework of interacting quantum field theory on causal sets. We derive a diagrammatic expansion for in-in correlators in local scalar field theories with finite polynomial interactions.
We outline how these same correlators can be computed using the double-path integral which acts as a generating functional for the in-in correlators. We modify the in-in generating functional to obtain a generating functional for in-out correlators. We define a notion of scattering amplitudes on causal sets with non-interacting past and future regions and verify that they are given by S-matrix elements (matrix elements of the time-evolution operator). We describe how these formal developments can be implemented to compute early universe observables under the assumption that spacetime is fundamentally discrete.}

\begin{document}
\maketitle
\flushbottom

\section{Introduction}
\label{sec:intro}

The challenge of quantum gravity is bridging the mathematical and conceptual disparities between quantum mechanics and general relativity \cite{Bambi:2023jiz,Carlip:2001wq,Oriti:2009zz,Sorkin:1997gi}. These disparities can also be credited for the proliferation of potential resolutions, since the resolution that we will reach will depend on the building-blocks we used to construct it (a canonical or path integral formalism, a continuum or a discretuum \textit{etc.}). But ultimately, this division between opposing principles should be reconciled as they all emerge from a unified theory of quantum gravity.

Causal set theory (CST) is an approach to quantum gravity in which spacetime is fundamentally discrete at the Planck scale \cite{surya2019causal,Dowker:2020qqs,Dowker:2013dog,Bambi:2023jiz,Yazdi:2023scl,Dowker:2011zz}. Taking its cue from theorems in Lorentzian geometry which state that in past- and future-distinguishing spacetimes (of which globally hyperbolic spacetimes are a subclass) the causal structure determines the metric uniquely up to the conformal factor \cite{hawking1976new,malament1977class}, CST elevates the causal structure of spacetime to be its principal feature \cite{bombelli1987space,bombelli1989origin,sorkin2005causal,surya2019causal}. The missing ingredient---the conformal factor or volume measure---is accounted for by the discreteness which replaces spacetime points with spacetime ``atoms'' which need only be counted in order to compute spacetime volumes. Concretely, in CST spacetime takes the form of a Lorentzian lattice or \textit{causal set}, a collection of elements arranged in a partial order $\prec$ which encodes the causal relations between them, where $x\prec y$ reads as ``$x$ precedes $y$'' or ``$x$ lies in the causal past of $y$''. The continuum geometries of general relativity must emerge from this discrete structure via some large-scale approximation and coarse-graining \cite{surya2019causal,Bombelli:2000wu}. Much progress has been made in recent years in extracting continuum quantities, such as dimension and curvature, directly from the causal set (see for example \cite{meyer1988dimension,reid2003manifold,roy2013discrete,benincasa2010scalar,sverdlov2009gravity,Eichhorn:2018doy}).

One of the important achievements of this line of work has been the formulation of quantum field theory on a causal set \cite{Johnston:2008za, Johnston:2009fr,Sorkin:2017fcp, Dable-Heath:2019sej,Sorkin:2011pn,albertini:2021,Jubb:2023mlv}. Acting as a bridge between the discrete and the continuum, this work was able to reproduce the Green functions of continuum geometries from underlying causal sets through combinatorial means \cite{X:2017jal,Johnston:2008za}. It also established a distinguished vacuum state on a causal set \cite{Sorkin:2017fcp}. This last achievement is revealing of the fact that a causal set is more akin to a curved geometry than to a flat one, sharing the challenges posed by curvature---there may be no distinguished vacuum state or no natural way to define asymptotic states or an S-matrix. Indeed, the construction of a distinguished vacuum on a causal set prompted the Sorkin-Johnston proposal for a distinguished vacuum state in continuum spacetimes \cite{afshordi2012ground,Afshordi:2012jf,aslanbeigi2013preferred,Surya:2018byh,Zhu:2022kcf} which coincides with the ground state of the Hamiltonian in the case of static spacetimes. It also offered new insight into the computation of entanglement entropy \cite{Yazdi:2022hhv,Chen:2020ild,Keseman:2021dkf}.

In this work, we build on \cite{albertini:2021,Jubb:2023mlv,Sorkin:2011pn} in extending the formalism of quantum field theory on causal sets to include interactions. We develop a systematic diagrammatic method for computing in-in correlators and define a notion of scattering on a causal set.

\paragraph{In-in formalism}
The in-in formalism is conceptually and mathematically well-suited to CST: it is well-known to be manifestly causal \cite{Dickinson:2013lsa,Musso:2006pt,Weinberg:2005vy}; it requires only the notions of a state and of local operators, both of which are readily available on a causal set; and it can be phrased in terms of the Keldysh-Schwinger double path integral, a close relative of the decoherence functional which describes spacetime dynamics in CST \cite{dowker2010extending, Surya:2020cfm}. These aspects come into play in our work and make the adaptation from the continuum to the discrete possible. In particular, we see direct analogues between our manifestly-causal diagrammatic framework for computing in-in correlators on a causal set with the work of \cite{Dickinson:2013lsa,Musso:2006pt} in the continuum. We also find one major difference: the fundamental discreteness of the causal set acts as a natural cut-off, eliminating the UV divergences which characterise the continuum (although IR divergences are present when the interacting region is infinite). This is because each diagram in our expansion is a sub-causal set of the interaction region. Since a finite interaction region contains only finitely many sub-causal sets, the associated expansion contains only finitely many diagrams and terminates at a finite order in the interaction coupling.

In recent years, the in-in formalism has become increasingly important as the modern sky surveys whose data it describes continue to make our cosmological observations increasingly precise (\textit{e.g.}, \cite{akrami2020planck,ballardini2016probing}), enabling us for the first time to glimpse into the early universe and probe regimes in which quantum gravity effects become important---the regimes in which CST predicts we will detect the Planckian discreteness of spacetime.
To make the most of these observational advances, we must develop robust and predictive theoretical frameworks through which the observations can be interpreted. This is one of the main motivations behind this work and our intention is for the construction which we present here to be applied in computing cosmological predictions for sky surveys (\textit{e.g.}, primordial non-gaussianities \cite{chen2010primordial,wang2014inflation}).

\paragraph{Discrete/continuum correspondence}
Our formalism enables the computation of in-in correlators on any causal set, but not every causal set is a physically sensible background. We are interested in those causal sets which can give rise to cosmological geometries through the so-called discrete/continuum correspondence \cite{Saravani:2014gza,surya2019causal,Johnston:2021lei}. We say that a causal set $(\mathcal{C},\prec)$ is well-approximated by a continuum geometry $(\mathcal{M},g)$ of dimension $d$ if there exists a \textit{faithful embedding} of $\mathcal{C}$ in $\mathcal{M}$, that is a map $f:\mathcal{C}\rightarrow\mathcal{M}$ such that,
\begin{itemize}
    \item[$(i)$] $q\prec p \iff f(q)\in J^-(f(p))$, where $J^-(x)$ denotes the causal past of $x\in\mathcal{M}$,
    \item[$(ii)$] the points $f(\mathcal{C})$ are distributed in $\mathcal{M}$ according to the
Poisson distribution at some fixed density $\rho$,
\item[$(iii)$] the discreteness length $l=\rho^{-\frac{1}{d}}$ is small compared to any curvature length scale in $\mathcal{M}$. 
\end{itemize}
Condition $(i)$ requires that $f$ preserves the causal order and condition $(ii)$ requires that the number of elements mapped by $f$ into an interval of
volume $V$ in $\mathcal{M}$ is approximately equal to $\rho V$. There are other ways in which one could imagine evenly distributing the points in $\mathcal{M}$, for example in a regular square grid. But in flat space such a grid picks a preferred frame: once the grid is embedded in the continuum, a Lorentz boost will render it no longer a regular grid. The Poisson process has the advantage of Lorentz invariance: a Lorentz transformation merely changes a Poisson distribution of points into another Poisson distribution with the same density \cite{bombelli2009discreteness,Saravani:2014gza,Dowker:2019oga}. 

``Sprinkling'' is a method for numerical implementation of the discrete/continuum correspondence: given a geometry $(\mathcal{M},g)$, one generates causal sets which faithfully embed into it by sampling points from $\mathcal{M}$ according to the Poisson distribution (\textit{i.e.}, by ``sprinkling'' into $\mathcal{M}$) \cite{Fewster:2020kqo,Saravani:2014gza}. This process is particularly suited to a cosmological setting where the observable spacetime volume is finite.

Given some combinatorial quantity $\textbf{Q}$ (such as a causal set in-in correlator), its average $\mathbb{E}_{\rho}(\textbf{Q})$ over a sample of causal sets produced via sprinkling into $\mathcal{M}$ at fixed $\rho$ is a quantity which we can associate with the continuum $(\mathcal{M},g)$. 
The result will depend on $\rho$, and the phenomenological effects of spacetime discreteness will be encoded in the high (but finite) $\rho$ corrections. In some cases, one can compute the ensemble average analytically and show that its $\rho\rightarrow\infty$ limit is exactly equal to an invariant of the continuum (\textit{e.g.}, this is the case for curvature invariants \cite{belenchia2016continuum}), suggesting that our framework may prove a useful computational tool independently of whether or not spacetime is fundamentally discrete. In particular, our framework may offer a novel regularization of the UV divergences of the continuum.

\paragraph{Scattering amplitudes} 
Scattering between asymptotic states is what we measure in particle accelerator experiments and the associated S-matrix has been an invaluable tool for studying quantum field theories on flat space \cite{Eden:1966dnq}. Its definition on de Sitter has also received some attention, where the lack of a conserved positive energy-like quantity poses an obstruction to the definition of asymptotic particle states \cite{Melville:2023kgd,Bousso:2004tv,Marolf:2012kh,Benincasa:2022gtd}. Therefore, it is natural to ask how one might define the S-matrix on a causal set. There are many continuum concepts and structures that are lacking on a causal set and in particular,  rather than continuum hypersurfaces labelled by $t$ we simply have a discrete label (taking its values in the positive integers) for each causal set site. These labels play the role of coordinates and their permutation (the analogue of a coordinate transformation) should leave the S-matrix invariant. This complicates the question of how analogues of asymptotic states should be defined. Should these states be associated with a specific site, and if so, which particular site would be appropriate? We suggest an answer to this question in the set-up where the interaction region is confined to an area between non-interacting ``past'' and ``future'' regions.

\paragraph{Main results} Our results are restricted to 
interacting scalar field theories on a causal set background, under the assumption that all interaction terms are local polynomials of a finite order. With this caveat, the main results of this work are as follows:
\begin{itemize}
\item \textit{A diagrammatic expansion for the Heisenberg field.} Our derivation of the expansion relies on the commutator representation of time evolution (attributed to Weinberg \cite{Weinberg:2005vy}) which in the continuum is equivalent to solving the full equation of motion for the Heisenberg field. We show that when the interaction region is finite, the expansion terminates at a finite order. Hence, in this scenario the Heisenberg field is a finite-order polynomial in the interaction picture fields and vice versa.
    \item \textit{A diagrammatic expansion for in-in correlators on a causal set.} Our expansion makes use of Wick's theorem and is therefore limited to the case when the in-state is a Gaussian state.  Aside from this caveat, our expansion is generic and can be used to compute ``non-equal time'' in-in correlators of composite operators (as there are no issues with coincident operators in the discrete). Our diagrams have two kinds of legs corresponding to the Feynman and the retarded propagator. Each interaction vertex must be connected to an external vertex via a path of retarded propagators, thus causality is manifest. When the interaction region is finite, the expansion terminates at a finite order.
    \item \textit{A generating functional for in-out correlators.} We review the generating functional for in-in correlators and modify it to obtain a generating functional for in-out correlators. These correlators have an infinite expansion even on a finite causal set.
    \item \textit{The S-matrix on a causal set.} We define the analogue of asymptotic states in the case where the interaction region is bounded between non-interacting past and future regions. We use this to define the S-matrix and give its diagrammatic expansion. The expansion contains the same diagrams as in the continuum (and hence does not terminate).

\end{itemize}

\paragraph{Outline} In Section \ref{sec:freetheory}, we summarise some useful causal set terminology and review the formalism for a free quantum field theory on a causal set. In Section \ref{sec: interacting fields}, we define the Heisenberg field in interacting theories in terms of the causal set time evolution operator and give its diagrammatic expansion. In Section \ref{sec: in in}, we give our diagrammatic rules for in-in correlators on a causal set (leaving the detailed derivation to Appendix \ref{in-in derivation appendix}). In Section \ref{sec: path int}, we outline how the in-in correlators of the previous section can be computed using a double-path integral whose measure is given by a decoherence functional. And we modify this generating functional to give a generating functional for in-out correlators. In Section \ref{sec: scattering} we define the S-matrix on a causal set.
We conclude with discussion of future directions in Section \ref{sec: conclusion}.

\section{Background}\label{sec:freetheory}
In this Section we introduce the causal set terminology which we will need and review the operator formalism for a free scalar quantum field theory on a causal set. For a review of causal set theory and related terminology see \cite{surya2019causal} and references therein.

\subsection{Causal sets}\label{subsec: causets}
A \emph{partially ordered set} or \emph{poset} is a set (called the \textit{underlying set}) together with an irreflexive, antisymmetric and transitive relation, denoted by $\prec$, on it. Given a poset and a pair of elements $x$ and $y$ in it which satisfy $x\prec y$, the associated \emph{interval}, $int(x,y)$, is,
\[
int(x,y)=\{z|x\prec z\prec y\}.
\]
A poset is \emph{locally finite} if every interval of the poset is finite. A \textit{causal set} or \textit{causet} is a locally finite poset. By the standard abuse of notation we will use $\mathcal{C}$ to denote both the causal set and its underlying set. Given a causet $\mathcal{C}$ and $x,y \in \mathcal{C}$ such that $x\prec y$, we say that $y$ \emph{covers} $x$ and write $x\precdot y$ if the interval $int(x,y)$ is empty (\textit{i.e.}, if there is no element $z$ which lies between $x$ and $y$ in the partial order). If $x$ and $y$ are unordered by $\prec$ we write  $x\natural y$. 

A Hasse diagram of a causet $\mathcal{C}$ is a graph whose vertices represent the elements of $\mathcal{C}$ and whose edges represent the covering relations in $\mathcal{C}$, where there is an upward-going edge from $x$ to $y$ if and only if $x \precdot y$ (the other relations are implied by transitivity). A Hasse diagram is the transitive reduction of a directed acyclic graph.

Using a Hasse diagram, we can think of a causal set as a Lorentzian lattice. The lattice sites play the role of spacetime points, while the causal structure is given by the direction of the edges: the past (future) of a lattice site $x$ are all the sites $y$ such that there exists an upward-going (downward-going) path from $y$ to $x$, and if a pair of sites $x$ and $y$ are such that there is no directed path
from one to the other then $x$ and $y$ are spacelike to each other. This is the interpretation of the causal set in CST: the partial order is the causal order, the past of an element $y$ are all elements $x$ such that $x\prec y$, the interval $int(x,y)$ plays the role of the Alexandrov set (the intersection of the past of a point with the future of another) \textit{etc.} It is in this sense the causal set encodes a causal structure and furnishes a discrete spacetime.

A tenet of CST is that the points (or ``atoms'') of the causal set spacetime are indistinguishable from one another, except through the partial order in which they are arranged. Mathematically, this idea is captured by order-isomorphism classes of causal sets, defined as follows. Two causets $\mathcal{C}$ and $\mathcal{C}'$ are \emph{order-isomorphic} if there exists a bijection $f:\mathcal{C}\rightarrow \mathcal{C}'$ that preserves the order relation, \textit{i.e.}  $f(x)\prec f(y) \Leftrightarrow x\prec y$ for all $x,y \in \mathcal{C}$ (where by standard abuse of notation we write $\prec$ to denote both the partial order in $\mathcal{C}$ and in $\mathcal{C}'$). Order-isomorphism is an equivalence relation and an \textit{unlabelled causet} is an order-isomorphism equivalence class. An unlabelled causet is represented by an unlabelled Hasse diagram and a causet representative of it can be constructed simply by labeling the vertices. In practice, working with labelled objects is simpler than working with unlabelled objects. Therefore, we will work with causets directly (not with their equivalence classes), bearing in mind that the choice of the underlying sets or of any labelling attached to the causet elements is pure gauge (just like coordinates are gauge in the continuum). For concreteness, we now specify the causets which we will work with. A causal set is \textit{finite} if its underlying set is finite, \textit{i.e.}, $|\mathcal{C}|<\infty$. A finite causal set $\mathcal{C}$ is \textit{naturally labelled} if its underlying set is $\{1,2,\ldots,|\mathcal{C}|\}$ and $x\prec y \implies x<y \ \forall \ x,y\in\mathcal{C}$. From here on, we will work with finite naturally labelled causal sets. If $\mathcal{C}$ and $\mathcal{C}'$ are naturally labelled causets which are order-isomorphic to each other, we will say that $\mathcal{C}$ is a \textit{relabelling} of $\mathcal{C}'$ (and vice versa), since the Hasse diagram of one can be obtained from the Hasse diagram of the other simply by relabelling the vertices. What we mean by saying that the labelling is pure gauge is that the physics remains unchanged under relabelling, \textit{i.e.}, under the interchange of $\mathcal{C}$ with $\mathcal{C}'$.

We will have use for the following causal set terminology. A subcauset $S\subseteq \mathcal{C}$ is called a \textit{stem} if $x\in S \implies y\in S$ for all $y\prec x$ in $\mathcal{C}$. A causet $\mathcal{C}$ is a \textit{chain} if it is totally ordered (\textit{i.e.}, $x\prec y$ or $y\prec x$ for all pairs $x,y\in\mathcal{C}$). 
A subcauset $S\subseteq \mathcal{C}$ is called a \textit{path} if it is a chain and $y\precdot x$ in $S \implies y\precdot x$ in $\mathcal{C}$.
An element $x\in\mathcal{C}$ is \textit{minimal} (\textit{maximal}) if there exists no $y\in\mathcal{C}$ such that $y\prec x$ ($y\succ x$).
We say that an element $x\in\mathcal{C}$ is in level $l$ if the longest chain of which $x$ is the maximal element has cardinality $l$. Thus, 
level 1 of $\mathcal{C}$ comprises the minimal elements, level 2 comprises the minimal elements of what remains of $\mathcal{C}$ after the
elements in level 1 are deleted, \textit{etc.}
\subsection{Free quantum field theory on a causal set}

We give a brief review of the Sorkin-Johnston construction for a free scalar field on a causal set \cite{Johnston:2008za,Johnston:2009fr,Sorkin:2011pn,Sorkin:2017fcp,X:2017jal}. This construction takes the retarded propagator as its starting point and uses it to write down covariant commutation relations for the field operators from which a distinguished vacuum state and its associated Fock representation can be derived. This approach is a covariant alternative to the canonical approach and is therefore suitable for our discrete setting. This is also the reason it is appropriate for curved (continuum) spacetimes to which it has been adapted in \cite{Afshordi:2012jf, aslanbeigi2013preferred}. When applied to spacetimes with a timelike Killing vector, the Sorkin-Johnston ground state coincides with the ground state of the Hamiltonian \cite{Afshordi:2012jf}.

A propagator is a real function of two spacetime points. On a (finite, naturally labelled) causet $\mathcal{C}$, we can represent a propagator as a matrix indexed by the elements of $\mathcal{C}$. The retarded propagator, denoted by $\Delta^R_{xy}$, is defined by the requirement that $\Delta^R_{xy}=0$ whenever $x\not\succ y$. Since $\mathcal{C}$ is naturally labelled, this constraint implies that $\Delta^R_{xy}$ is lower-triangular. Various prescriptions can be found in the literature for fixing its non-vanishing entries. These prescriptions can be roughly split into two camps: those in which $\Delta^R_{xy}$ is defined as the inverse of a suitably discretised d'Alembertian \cite{Dable-Heath:2019sej, Aslanbeigi:2014zva} and those in which $\Delta^R_{xy}$ is defined in a way which guarantees that it approaches the continuum retarded propagator in an appropriate limit \cite{X:2017jal,Johnston:2008za}. Depending on the exact prescription, $\Delta^R_{xy}$ may depend on parameters such as the dimension of the continuum spacetime which approximates $\mathcal{C}$, but once these parameters are fixed we may regard $\Delta^R_{xy}$ as purely combinatorial in the sense that it can be obtained directly from $\mathcal{C}$ (\textit{e.g.}, via its adjacency matrix). What follows is independent of the choice of $\Delta^R_{xy}$.

Given the retarded propagators, we define the advanced propagator and the Pauli-Jordan (also known as the causal propagator), respectively, as,
\begin{equation}
    \begin{split}
    &\Delta^A_{xy}=\Delta^R_{yx},\\
    &\Delta_{xy}=\Delta^R_{xy}-\Delta^A_{xy}.
        \end{split}
\end{equation}
Next, associate a field operator $\phi(x)$ to each $x\in\mathcal{C}$ and impose the Peierls bracket,
\begin{equation}\label{peierls bracket}
    [\phi(x),\phi(y)]=i\Delta_{xy}.
\end{equation}
Note that Peierls brackets guarantees that $\phi(x)$ and $\phi(y)$ commute if $x$ and $y$ are spacelike separated (\textit{i.e.}, if $x\natural y$). Following from its definition, the matrix $i\Delta_{xy}$ is Hermitian and can be decomposed via the spectral theorem as,
\begin{equation}\label{spectral decomposition eq}\begin{split}
     i\Delta_{xy}&=\sum_{\lambda} \lambda v^{(\lambda)}_x\bar{v}^{(\lambda)}_y\\
     &=\sum_{\lambda>0} \lambda v^{(\lambda)}_x\bar{v}^{(\lambda)}_y-\overline{\sum_{\lambda>0} \lambda v^{(\lambda)}_x\bar{v}^{(\lambda)}_y},
\end{split}
\end{equation}
where the bars denote complex conjugation and the $v^{(\lambda)}_x$ and the $\lambda$ are the eigenvectors and eigenvalues of $i\Delta_{xy}$, respectively, 
\begin{equation}
    i\Delta_{xy}v^{(\lambda)}_x=\lambda v^{(\lambda)}_x.
\end{equation} 
In the first line of \eqref{spectral decomposition eq}, the sum is over all the eigenvalues $\lambda$. The second line (where the sums are restricted to the positive $\lambda$) is obtained by using the fact that the $\lambda$'s are real and that if $\lambda>0$ is an eigenvalue with eigenvector $v^{(\lambda)}$ then $-\lambda$ is an eigenvalue with eigenvector $v^{(-\lambda)}=\bar{v}^{(\lambda)}$. The first term on the second line is called the \textit{positive part} of $i\Delta$,
\begin{equation}
    Pos(i\Delta)=\sum_{\lambda>0} \lambda v^{(\lambda)}_x\bar{v}^{(\lambda)}_y,
\end{equation}
while the second term is its complex conjugate. 
Another way to express $i\Delta$ as a difference between a c-number and its complex conjugate is to consider the expectation value of the Peierls bracket in some state of choice to obtain,
\begin{equation}
    \langle[\phi(x),\phi(y)]\rangle=W_{xy}-\overline{W}_{xy}=i\Delta_{xy},
\end{equation}
where $W_{xy}=\langle \phi(x)\phi(y)\rangle$ is the Wightman function in the chosen state.
The insight of the Sorkin-Johnston formalism is that the eigenvectors of $i\Delta$ form a distinguished basis which can be used to define a distinguished Gaussian vacuum state $|0\rangle$ by requiring that,
\begin{equation}\label{eq sj W}
\langle0|\phi(x)\phi(y)|0\rangle=Pos(i\Delta_{xy}).
\end{equation}
We call $|0\rangle$ the Sorkin-Johnston (SJ) vacuum.

The SJ vacuum has a Fock representation where the positive eigenvalues and their associated eigenvectors play the role of positive frequencies and mode functions, respectively. For each $\lambda>0$, introduce a pair of conjugate ladder operators $a_{\lambda}$ and $a_{\lambda}^{\dagger}$, and impose the commutation relations,
\begin{equation}\label{eq: comm rel a}
[a_{\lambda},a_{\lambda'}^{\dagger}]=\delta_{\lambda,\lambda'}.
\end{equation}
The SJ vacuum is the state annihilated by the $a_{\lambda}$ and the Fock representation is built by acting on it with the $a_{\lambda}^{\dagger}$. Expanding the fields as,
\begin{equation}\label{free field mode expansion}
    \phi(x)=\sum_{\lambda>0} \sqrt{\lambda}\bigg( v^{\lambda}_x a_{\lambda}+\bar{v}^{\lambda}_x a_{\lambda}^{\dagger}\bigg),
\end{equation}
one can verify that the Wightman function of the SJ state is given by \eqref{eq sj W}.

\subsection{Causal ordering and the Feynman propagator}\label{sec label order}
To complete our discussion of the free theory, we introduce the notion of \textit{causal ordering}. Given some $x,y\in\mathcal{C}$, we say that the product $\phi(x)\phi(y)$ is causally ordered if $x\not\prec y$. We write $C$ to denote the causal ordering operator whose action on a product of two fields is,
\begin{equation} C[\phi(x)\phi(y)]=
    \begin{cases}
      \phi(x)\phi(y) & \text{if $x\succ y$}\\
       \phi(y)\phi(x) & \text{if $x\prec y$},\\
    \end{cases}\end{equation}
where for a spacelike pair of points $x\natural y$ the associated field operators commute and the $C$ ordering is trivial: $C[\phi(x)\phi(y)]=\phi(x)\phi(y)=\phi(y)\phi(x)$. More generally, a product of field operators is causally ordered if no field operator is to the right of an operators which lives in its past. In a labelled causal set, ordering a product of operators by decreasing label from left to right always results in a causal ordering, \textit{e.g.} $\phi(4)\phi(4)\phi(2)\phi(1)$.

The anti-causal ordering operator $\overline{C}$ orders a product of field operators so that no field operator is to the left of an operators which lives in its past. In a labelled causal set, ordering a product of operators by increasing label from left to right always results in an anti-causal ordering, \textit{e.g.} $\phi(1)\phi(2)\phi(4)\phi(4)$.

Causal ordering is the causal set analogue of the time ordering of the continuum, and we use it to define the Feynman propagator,
\begin{equation}
    \Delta^F_{xy}=\langle C[\phi(x)\phi(y)]\rangle.
\end{equation}

\section{Interacting fields}\label{sec: interacting fields}

In the interacting theory, the interaction picture fields carry the free evolution. In the causal set context, this means that the interaction picture fields satisfy the Peierls bracket \eqref{peierls bracket} and have the mode expansion \eqref{free field mode expansion}. We will denote the interaction picture fields simply by $\phi(x)$. We write $\mathcal{H}$ to denote the interacting Hamiltonian (density) in the interaction picture, and we restrict ourselves to local Hamiltonians which are finite polynomials in the field, \textit{e.g.} $\mathcal{H}(x)=\frac{g}{n!}\phi^n(x)$.

The Heisenberg picture is defined by analogy with the continuum, where the Heisenberg field $\phi^H(t,\textbf{x})$ is related to the interaction picture field $\phi(t,\textbf{x})$ via,
\begin{equation}
    \phi^H(t,\textbf{x})=U^{\dagger}(t,t_0)\phi(t,\textbf{x}) U(t,t_0).
\end{equation}
where,
\begin{equation}
    U(t,t_0)=T\bigg[e^{-i \int_{t_0}^{t} H(t) dt}\bigg] \ \ \ (t\geq t_0)
\end{equation}
is the time-evolution operator and where $H$ is the interacting Hamiltonian in the interaction picture. The analogy suggests that we replace the time integral by a sum over causal set points and the time-ordering $T$ with the causal ordering $C$ (defined in Section \ref{sec label order}). Under the action of $C$, all field commutators vanish and we can express the exponential of a sum as a product of exponentials. This is the motivation for the following definitions.

Given a causet $\mathcal{C}$ and some $x,y\in\mathcal{C}$, we define the following family of evolution operators,
\begin{align}
\begin{split}\label{def V}
   &   V_x= C\bigg[\prod_{z\prec x} e^{-i\mathcal{H}(z)}\bigg],
\end{split}\\
  \begin{split}\label{def Uxy}
    & U_{x,y}=C\bigg[\prod_{y\leq  z < x} e^{-i\mathcal{H}(z)}\bigg]\text{ for }x>y,
    \end{split}\\
    \begin{split}\label{def U}
 &U_{x}=\begin{cases}
        1 & \text{if $x=1$,}\\
       U_{x,1}=C\bigg[\prod_{z < x} e^{-i\mathcal{H}(z)}\bigg] & \text{if $x>1$},\\
    \end{cases}
\end{split}
\end{align}
and note that they satisfy the following relations,
\begin{align}
\begin{split}
   &   V^{\dagger}_x V_x=U^{\dagger}_x U_x=1,
\end{split}\\
    \begin{split}\label{UU composition}
 &U_{x,y}U_{y,z}=U_{x,z},
\end{split}\\
\begin{split}\label{UUdagger composition}
U_{x,z}U_{y,z}^{\dagger} =
    \begin{cases}
      1 & \text{if $x=y$}\\
       U_{x,y} & \text{if $x>y$}\\
       U_{y,x}^{\dagger} & \text{if $x<y$},\\
    \end{cases}  \end{split}\\ 
    \begin{split}\label{U V equiv}
   & V_x^{\dagger}\mathcal{O}(x) V_x=U_x^{\dagger} \mathcal{O}(x) U_x \text{ for any local operator }\mathcal{O}(x).
    \end{split}
\end{align}
We say that $V_x$ is a covariant operator because it relies only the partial order $\prec \ $, while $U_{x,y}$ and $U_{x}$ are label-dependent operators because they rely on the total order $<$ of the labelling. In other words, $V_x$ is a physical operator and $U_{x,y}$ and $U_{x}$ are gauge-dependent operators.

We use the covariant operator $V_x$ to define the Heisenberg picture on a causal set: given a local operator $\mathcal{O}(x)$ in the interaction picture, the Heisenberg picture operator $\mathcal{O}^H(x)$ is given by, \begin{equation}\label{Heisenberg}
    \mathcal{O}^H(x)=V_x^{\dagger} \mathcal{O}(x)V_x.
\end{equation}
In practice, working with gauge-dependent operators is simpler and thanks to relation \eqref{U V equiv} we can express the Heisenberg operator in terms of $U_x$ as,
\begin{equation}\label{Heisenberg U}
    \mathcal{O}^H(x)=U_x^{\dagger} \mathcal{O}(x)U_x.
\end{equation}
(Note that our convention for the causal set evolution operators differs from that of \cite{albertini:2021,Jubb:2023mlv}, but the resulting Heisenberg operators are independent of convention.)

\subsection{Expanding the Heisenberg field}\label{subsec: expand H field}
In the continuum, one solves the full equation of motion order by order in the interaction coupling to obtain a perturbative expansion for the Heisenberg field in terms of interaction picture fields. Alternatively, the same expansion can be extracted from the nested commutator expression for the Heisenberg field \cite{Weinberg:2005vy} (see also \cite{Musso:2006pt,Dickinson:2013lsa}),
   \begin{equation}\label{continuum heisenberg master eq}\begin{split}
\phi^H(t,\textbf{x}) =& \sum_{n=0}^{\infty}  (-i)^n\int_{-\infty}^{t} dt_1\int_{-\infty}^{t_1} dt_2\ldots \int_{-\infty}^{t_{n-1}}dt_n \ \bigg[\ldots \bigg[\bigg[\phi(t,\textbf{x}),H(t_1)\bigg],H(t_2)\bigg]\ldots,H(t_n)\bigg],\\
=&\sum_{n=0}^{\infty} (-i)^n\int_{-\infty}^{t} dt_1\ldots \int_{-\infty}^{t}dt_n \ \bigg[\ldots \bigg[\bigg[\phi(t,\textbf{x}),{H}(t_1)\bigg],{H}(t_2)\bigg]\ldots,{H}(t_n)\bigg]\Theta(t_1,\ldots,t_n),\\
    \end{split}
\end{equation}
where in the second line the time integrals all share the same domain $(-\infty,t)$ thanks to the introduction of the generalised step function $\Theta(t_1,\ldots,t_n)$ defined as,
\begin{equation}
  \Theta(t_1,\ldots,t_n)= \begin{cases}
      1 & t_1\geq t_2\geq \cdots\geq t_n\\
      0 & \text{otherwise.}\\
    \end{cases}    
\end{equation}
Equivalently, $\Theta(t_1,\ldots,t_n)$ can be expressed as a product of step functions,
\begin{equation}
  \Theta(t_1,\ldots,t_n)=  \Theta(t_1-t_2)\ \Theta(t_2-t_3)\cdots \Theta(t_{n-1}-t_n),
\end{equation}
where, \begin{equation}
    \Theta(x-a)= \begin{cases}
      1 & x\geq a\\
      0 & x< a. \\
    \end{cases}    
\end{equation}
(This latter form makes explicit the similarities with the retarded product of \cite{Hagedorn:1961gj}.)

The above formalism can be adapted to the causal set setting where it provides a route for expanding the Heisenberg field in the absence of an equation of motion. As we show in Appendix \ref{in-in derivation appendix}, starting from expression \eqref{Heisenberg U} for the causal set Heisenberg field one obtains the following expansion,
    \begin{equation}\label{eq: Heisenberg n order discrete}\begin{split}
\phi^H(x)=\sum_{n=0}^{\infty}  (-i)^n\sum_{z_1=1}^{x-1}\cdots \sum_{z_n=1}^{x-1}  \bigg[\ldots\bigg[\bigg[\phi(x),\mathcal{H}(z_1)\bigg],\mathcal{H}(z_2)\bigg]\ldots,\mathcal{H}(z_n)\bigg]\Upsilon(z_1,\ldots, z_n) ,\\
\end{split}
\end{equation}
where $\phi(x)$ is the interaction picture field, $\mathcal{H}(z)$ is the interacting Hamiltonian in the interaction picture and $\Upsilon(z_1,\ldots, z_n)$ is a modification of $\Theta(z_1,\ldots, z_n)$ which is sensitive the saturation of the inequality $z_1\geq z_2 \geq \ldots\geq z_n$,
\begin{equation}
  \Upsilon(z_1,\ldots, z_n)= \begin{cases}
      \frac{1}{l_1!\ldots l_k!} & z_1\geq z_2 \geq \ldots \geq z_n\\
      0 & \text{otherwise,}\\
    \end{cases}    
\end{equation}
where $k$ denotes the number of strings of equal signs in $z_1\geq z_2 \geq \ldots \geq z_n$ and $l_i$ is the length of the $i^{th}$ string of equal signs (\textit{i.e.}, the number of $z$ variables equated by the string). For example, $\Upsilon(8,7,5,2,1)=1, \Upsilon(7,7,5,4,4,4)=\frac{1}{2!1!3!}$ and $\Upsilon(7,5,7,4,4,4)=0$.

(Note that one could replace $\Theta$ with $\Upsilon$ in \eqref{continuum heisenberg master eq} since $\Theta$ and $\Upsilon$ differ only on sets of measure zero. On the other hand, in the discrete there are no non-empty sets of measure zero. In particular, the cases when a number of variables take the same value contribute non-vanishingly to the sum and the definition of $\Upsilon$ ensures that these contributions are weighted correctly.)

To evaluate \eqref{eq: Heisenberg n order discrete}, we apply the result of \cite{Transtrum:2005} which expresses the commutator of functions of non-commuting operators as an expansion in powers of the commutators of these operators with each other. In our case, the operators are the fields $\phi(1),\ldots, \phi(x-1)$ which satisfy the commutation relations \eqref{peierls bracket}. Following \cite{Dickinson:2013lsa} we define,
\begin{equation}
\mathcal{F}_n(\phi(x),\phi(z_1),\ldots,\phi(z_n))=\bigg[\ldots\bigg[\bigg[\phi(x),\mathcal{H}(z_1)\bigg],\mathcal{H}(z_2)\bigg]\ldots,\mathcal{H}(z_n)\bigg],
\end{equation}
and note that $\mathcal{F}_n=\bigg[\mathcal{F}_{n-1},\mathcal{H}(z_n)\bigg]$. This enables us to apply the result of \cite{Transtrum:2005} as, 
\begin{equation}\begin{split} \label{commutator_identity}
\mathcal{F}_n=& - \underbrace{\sum_{k_{1}}\cdots \sum_{k_{n-1}}\sum_{k_{x}}}_{k\equiv \sum k_i \not=0} \bigg( \frac{(-i\Delta_{x,z_n})^{k_x}}{k_x!}\prod_{i=1}^{n-1} \frac{(-i\Delta_{z_{i},z_n})^{k_i}}{k_i!}\bigg) \partial_{\phi(z_1)}^{k_1}\cdots\partial_{\phi(z_{n-1})}^{k_{n-1}}\partial_{\phi(x)}^{k_x} \mathcal{F}_{n-1} \partial_{\phi(z_n)}^k \mathcal{H}(z_n),\\
\end{split}
\end{equation}
where the sums are over all non-negative integers with the restriction that at least one of the $k_i$ is non-zero. This expansion has a useful diagrammatic representation which we give below.

\paragraph{The diagrammatic expansion for the nested commutator}\label{diagrammatic expansion for the nested commutator} 
Associate a vertex with each of the points $x,z_1\ldots ,z_n$. We call $x$ the external vertex and $z_i$ the internal vertices. The number of half-legs meeting at each vertex is equal to the number of fields at the associated point. To form the diagrams, connect the half-legs in all possible ways to form directed edges with the following properties: $(i)$ every internal vertex is connected to the external vertex by at least one directed path, and $(ii)$ every directed edge is of the form $z_i\rightarrow x$ or $z_i\rightarrow z_j$ with $i>j$. Property $(i)$ corresponds to the restriction that at least one of the $k_i$ is non-zero. Property $(ii)$ reflects the factors of $\Delta_{x,z_n}$ and $\Delta_{z_i,z_n}$ in \eqref{commutator_identity}, in particular the fact that $z_n$ is always the second argument. Connecting the legs corresponds to the taking of derivatives and each possible diagram corresponds to a different set of values for $k_1,\ldots, k_{n-1},k_x$. The nested commutator in \eqref{eq: Heisenberg n order discrete} is equal to the sum of these diagrams when each diagram is interpreted as follows:
\begin{itemize}
\item each directed edge $a\rightarrow b$ gives a factor of $-i\Delta_{ba}$,
\item each internal vertex contributes a factor of the coupling and any other constants appearing in $\mathcal{H}$,
\item each uncontracted half-leg at a vertex contributes a field operator at the associated point; the product of these factors is ordered with $\phi(x)$ on the left and then factors of $\phi(z_i)$ with $i$ increasing from left to right,
\item multiply by an overall factor coming from the different ways of connecting the half-legs to form the diagram,
\item multiply by an overall factor of $(-1)^n$.
\end{itemize}
For example, when $\mathcal{H}=\frac{g}{3!}\phi^3$ the Heisenberg field is given by,
\begin{equation}\label{phi3 example}
\begin{split}
 \phi^H(x) =&
 \begin{tikzpicture}[baseline=(current bounding box.center), decoration={markings,
    mark= at position 0.6 with {\arrow{Straight Barb[ length=1mm, width=1mm]}}}
] 
  \tikzstyle{every node}=[circle, draw, fill=black,inner sep=0pt, minimum width=3pt]
  \node (n3) at (9,0) [label= right: $x$]{};
    \draw [](9,-0.5) -- (n3)  {} ;
\end{tikzpicture}
 -i \sum_{z_1=1}^{x-1}
\begin{tikzpicture}
[baseline=(current bounding box.center),decoration={markings, 
    mark= at position 0.8 with {\arrow{Straight Barb[ length=1mm, width=1mm]}}}
] 
  \tikzstyle{every node}=[circle, draw, fill=black,inner sep=0pt, minimum width=3pt]
  \node (n3) at (9,0.5) [label= right: $x$]{};
  \node(n2) at (9,0)  [label= right: $z_1$]{};
    \draw [postaction={decorate}](n2) -- (n3)  {} ;
    \draw [](8.5,-0.5) -- (n2)  {} ;
     \draw [](9.5,-0.5) -- (n2)  {} ;
\end{tikzpicture}
+(-i)^2\sum_{z_1,z_2=1}^{x-1}\left(
\begin{tikzpicture}
[baseline=(current bounding box.center),decoration={markings, 
    mark= at position 0.6 with {\arrow{Straight Barb[ length=1mm, width=1mm]}}}
] 
  \tikzstyle{every node}=[circle, draw, fill=black,inner sep=0pt, minimum width=3pt]
  \node[baseline](n3) at (9,0.5) [label= right: $x$]{};
  \node (n2) at (9,0)  [label= right: $z_1$]{};
   \node (n1) at (9,-0.5)  [label= right: $z_2$]{};
    \draw [postaction={decorate}](n2) -- (n3)  {} ;
    \draw [] (8.5,-0.5) -- (n2)  {} ;
 \draw [](8.5,-1) -- (n1)  {} ;
     \draw [](9.5,-1) -- (n1)  {} ;
     \draw [postaction={decorate}](n1) -- (n2)  {} ;
\end{tikzpicture}
+
\begin{tikzpicture}
[baseline=(current bounding box.center),decoration={markings, 
    mark= at position 0.6 with {\arrow{Straight Barb[ length=1mm, width=1mm]}}}
] 
  \tikzstyle{every node}=[circle, draw, fill=black,inner sep=0pt, minimum width=3pt]
  \node[align=left] (n3) at (9,1) [label= right: $x$]{};
  \node[align=left] (n2) at (9,0.5)  [label= right: $z_1$]{};
   \node[align=left]  (n1) at (9,0)  [label= right: $z_2$]{};
    \draw [postaction={decorate}](n2) -- (n3)  {} ;
    \draw [postaction={decorate}](n1) to[bend right] (n2)  {} ;
     \draw [postaction={decorate}](n1) to[bend left] (n2)  {} ;
     \draw [](9,-0.5) -- (n1)  {} ;
\end{tikzpicture}\right) \Upsilon(z_1, z_2)+\cdots
\nonumber
\\
=& \phi(x)+\frac{g}{2} \, \sum_{z_1=1}^{x-1}\Delta^R_{x,z_1}\phi(z_1)^2 +\frac{g^2}{2}\sum_{z_1,z_2=1}^{x-1}\Delta^R_{x,z_1}\Delta^R_{z_1,z_2}\big( \phi(z_1)\phi(z_2)^2-2i\Delta^R_{z_1,z_2}\phi(z_2)\big)\Upsilon(z_1, z_2)+\cdots,\\
\end{split}
\end{equation}
where in the second line we used the relation, \begin{equation}
    \Delta_{z_i z_j}\Upsilon(z_1,\ldots,z_n)=\Delta^R_{z_i z_j}\Upsilon(z_1,\ldots,z_n)
\end{equation} to replace the Pauli-Jordan $\Delta$ by the retarded propagator $\Delta^R$. In this work, nested commutators always appear together with the ordering function $\Upsilon$ and therefore we always interpret the directed edges as retarded propagators.

The field expansion \eqref{eq: Heisenberg n order discrete} terminates at a finite order in the interaction coupling. This order increases with the order of the interaction Hamiltonian and with the level of $x$ (the length of longest chain $x$ is a maximal element of, see Section \ref{subsec: causets}).
This becomes particularly transparent when considering the diagrammatic representation for $\phi^H(x)$. Suppose $x$ is at level 3 with $x\succ x_1\succ x_2$, and let $\mathcal{H}=\frac{g}{4!}\phi^4$. The diagram which contains the largest number of interaction vertices is, \begin{equation}\begin{tikzpicture}
[baseline=(current bounding box.center),decoration={markings, 
    mark= at position 0.6 with {\arrow{Straight Barb[ length=1mm, width=1mm]}}}
] 
  \tikzstyle{every node}=[circle, draw, fill=black,inner sep=0pt, minimum width=3pt]
  \node (n3) at (9,0.7) [label= right: $x$]{};
  \node(n2) at (9,0)  [label= right: $x_1$]{};
   \node(n1) at (9,-0.7)  [label= right: $x_2$]{};
   \node(n0) at (8.3,-0.7)  [label= left: $x_2$]{};
   \node(nx) at (9.8,-0.7)  [label= right: $x_2$]{};
    \draw [postaction={decorate}](n2) -- (n3)  {} ;
    \draw [postaction={decorate}](n1)  -- (n2)  {} ;
    \draw [postaction={decorate}](n0)  -- (n2)  {} ;
    \draw [postaction={decorate}](nx)  -- (n2)  {} ;
    
     \draw [](9.2,-1.4)  --  (n1)  {} ;
     \draw [](9,-1.4)  --  (n1)  {} ;
     \draw [](8.8,-1.4)  --  (n1)  {} ;

     \draw [](8.1,-1.4)  --  (n0)  {} ;
     \draw [](8.3,-1.4)  --  (n0)  {} ;
     \draw [](8.5,-1.4)  --  (n0)  {} ;
     
     \draw [](10,-1.4)  --  (nx)  {} ;
     \draw [](9.8,-1.4)  --  (nx)  {} ;
     \draw [](9.6,-1.4)  --  (nx)  {} ;
\end{tikzpicture}=\frac{g^4}{216} \Delta^R_{x,x_1}(\Delta^R_{x_1,x_2})^3\phi(x_2)^9,\end{equation}
where $z_1=x_1$ and $z_2=z_3=z_4=x_2$. Since $x_2$ is a minimal element, there is no way to add another internal vertex to produce a non-vanishing diagram. This illustrates the general pattern: when $\mathcal{H}\sim \phi^r$, the diagrams with the highest number of interaction vertices are those in which the directed edges form a tree with 1 interaction vertex as the root, $r-1$ interaction vertices directly below it, $(r-1)^2$ interaction vertices below them \textit{etc.} with the last layer containing $(r-1)^{l-2}$ vertices, where $l$ denotes the level of $x$. Hence, the expansion terminates at order $\sum_{k=0}^{l-2}{(r-1)^{k}}=\frac{(r-1)^{l-1}-1}{r-2}$ in the coupling constant (which is of order $(r-1)^{l-1}$ in the field, \textit{i.e.} $\phi^{(r-1)^{l-1}}$).

\subsection{Properties of the field algebras}\label{sec:algebras}
Here we list some properties of the Heisenberg and interaction picture fields on a causal set.
\paragraph{Causality.} Heisenberg fields at spacelike separated points commute,
\begin{equation}
    [\phi^H(x),\phi^H(y)]=0 \text{ for all }x,y\in \mathcal{C} \text{ with } x\natural y.
\end{equation}
This follows from the definition \eqref{Heisenberg U} of the Heisenberg field and the Peierls brackets \eqref{peierls bracket}.

\paragraph{Polynomial property.} As we saw in Section \ref{subsec: expand H field}, the Heisenberg field $\phi^H(x)$ can be written as, \begin{equation}\label{polynomial property H field}
    \phi^H(x)=\phi(x)+Q_x(\phi(y); y\prec x),
\end{equation}
where $Q_x$ is a finite order polynomial in the interaction picture fields in the past of $x$. We can also invert this relationship and write,
 \begin{equation}\label{polynomial property interaction field}
    \phi(x)=\phi^H(x)+P_x(\phi^H(y); y\prec x),
\end{equation}
where $P_x$ is a finite order polynomial in the Heisenberg fields in the past of $x$. To see this, suppose \eqref{polynomial property interaction field} is true for all $x$ in levels $0,1,\ldots, l$ and consider $x$ at level $l+1$. Then $Q_x$ is a polynomial in fields $\phi(y)$ with $y$ at level $l$ or below. By the inductive assumption, these $\phi(y)$ can be written as a finite order polynomial in Heisenberg picture fields. Defining $P_x(\phi^H(y); y\prec x)=-Q_x(\phi(y); y\prec x)$ and rearranging \eqref{polynomial property H field} for $\phi_x$ gives the result.

\paragraph{Observable algebras.} Given a causet $\mathcal{C}$ and a subcauset $\mathcal{R}\subseteq \mathcal{C}$, we write $\mathfrak{A}_\mathcal{R}^H$ and $\mathfrak{A}_\mathcal{R}$ denote the algebras generated by $\{\phi^H(x)\}_{x\in\mathcal{R}}$ and $\{\phi(x)\}_{x\in\mathcal{R}}$, respectively. 
Then it is a corollary of the polynomial property that $\mathfrak{A}_\mathcal{R}^H=\mathfrak{A}_\mathcal{R}$ if and only if $\mathcal{R}$ is a stem in $\mathcal{C}$.

\section{The in-in formalism}\label{sec: in in}
Here we build on the results of Section \ref{subsec: expand H field} to obtain a diagrammatic expansion for in-in correlators on a causal set. Throughout, we assume that the in-state is a Gaussian state so Wick's theorem applies.

\subsection{The expectation value of the field}\label{subsec: field exp value}
We begin with the simplest in-in correlator: the expectation value of the field, $ \langle \phi^H(x) \rangle$. Having expanded $\phi^H(x)$ in Section \ref{subsec: expand H field} as a causally ordered polynomial in the interaction picture fields, we can apply Wick's theorem to obtain $ \langle \phi^H(x) \rangle$. Diagrammatically, this corresponds to taking each diagram in the expansion of $\phi^H(x)$ and joining its half-legs to create undirected edges in all possible ways.  Each undirected edge is a Feynman propagator, and each diagram is weighted by an additional Wick factor. Assuming $\langle \phi(x) \rangle=0$, only the diagrams in which all half-legs are contracted contribute to the expectation value $\langle \phi^H(x) \rangle$.
For example, when $\mathcal{H}=\frac{g}{3!}\phi^3$ we apply Wick's theorem to \eqref{phi3 example} and obtain,\begin{equation}\label{exp phi3 example}
\begin{split}
 \langle\phi^H(x)\rangle =&
 -i \sum_{z_1=1}^{x-1}
\begin{tikzpicture}
[baseline=(current bounding box.center),decoration={markings, 
    mark= at position 0.6 with {\arrow{Straight Barb[ length=1mm, width=1mm]}}}
] 
  \tikzstyle{every node}=[circle, draw, fill=black,inner sep=0pt, minimum width=3pt]
  \node (n3) at (9,0.7) [label= right: $x$]{};
  \node(n2) at (9,0)  [label= right: $z_1$]{};
    \draw [postaction={decorate}](n2) -- (n3)  {} ;
  \draw (n2) to [loop below, scale=1.5, looseness=20,min distance=5mm] (n2){};
\end{tikzpicture} 
\\+&
(-i)^3 \sum_{z_1,z_2,z_3=1}^{x-1}
\left(\begin{tikzpicture}
[baseline=(current bounding box.center),decoration={markings, 
    mark= at position 0.6 with {\arrow{Straight Barb[ length=1mm, width=1mm]}}}
] 
  \tikzstyle{every node}=[circle, draw, fill=black,inner sep=0pt, minimum width=3pt]
  \node (n3) at (9,0.7) [label= right: $x$]{};
  \node(n2) at (9,0)  [label= right: $z_1$]{};
   \node(n1) at (9,-0.7)  [label= right: $z_2$]{};
    \node(n0) at (9,-1.4)  [label= right: $z_3$]{};
    \draw [postaction={decorate}](n2) -- (n3)  {} ;
    \draw [postaction={decorate}](n1) to[bend left] (n2)  {} ;
    \draw (n1)to[bend right] (n2)  {} ;
    \draw [postaction={decorate}](n0) -- (n1)  {} ;
  \draw (n0) to [loop below, scale=1.5, looseness=5,min distance=5mm] (n0){};
\end{tikzpicture}+\begin{tikzpicture}
[baseline=(current bounding box.center),decoration={markings, 
    mark= at position 0.6 with {\arrow{Straight Barb[ length=1mm, width=1mm]}}}
] 
  \tikzstyle{every node}=[circle, draw, fill=black,inner sep=0pt, minimum width=3pt]
  \node (n3) at (9,0.7) [label= right: $x$]{};
  \node(n2) at (9,0)  [label= right: $z_1$]{};
   \node(n1) at (9,-0.7)  [label= right: $z_2$]{};
    \node(n0) at (9,-1.4)  [label= right: $z_3$]{};
    \draw [postaction={decorate}](n2) -- (n3)  {} ;
    \draw [postaction={decorate}](n1) to[bend left] (n2)  {} ;
    \draw [postaction={decorate}](n1)to[bend right] (n2)  {} ;
    \draw [postaction={decorate}](n0) -- (n1)  {} ;
  \draw (n0) to [loop below, scale=1.5, looseness=5,min distance=5mm] (n0){};
\end{tikzpicture}+
\begin{tikzpicture}
[baseline=(current bounding box.center),decoration={markings, 
    mark= at position 0.6 with {\arrow{Straight Barb[ length=1mm, width=1mm]}}}
] 
  \tikzstyle{every node}=[circle, draw, fill=black,inner sep=0pt, minimum width=3pt]
  \node (n3) at (9,0) [label= right: $x$]{};
  \node(n2) at (9,-0.7)  [label= right: $z_1$]{};
   \node(n1) at (8.8,-1.4)  [label= left: $z_2$]{};
    \node(n0) at (9.2,-1.4)  [label= right: $z_3$]{};
    \draw [postaction={decorate}](n2) -- (n3)  {} ;
    \draw [postaction={decorate}](n1) -- (n2)  {} ;
     \draw [postaction={decorate}](n0)-- (n2)  {} ;
  \draw (n1) to [loop below, scale=1.5, looseness=5,min distance=5mm] (n1){};
  \draw (n0) to [loop below, scale=1.5, looseness=5,min distance=5mm] (n0){};
\end{tikzpicture}+
\begin{tikzpicture}
[baseline=(current bounding box.center),decoration={markings, 
    mark= at position 0.6 with {\arrow{Straight Barb[ length=1mm, width=1mm]}}}
] 
  \tikzstyle{every node}=[circle, draw, fill=black,inner sep=0pt, minimum width=3pt]
  \node (n3) at (9,0.7) [label= right: $x$]{};
  \node(n2) at (9,0)  [label= right: $z_1$]{};
   \node(n1) at (9,-0.7)  [label= right: $z_2$]{};
    \node(n0) at (9,-1.4)  [label= right: $z_3$]{};
    \draw [postaction={decorate}](n2) -- (n3)  {} ;
    \draw [postaction={decorate}](n1) -- (n2)  {} ;
     \draw [postaction={decorate}](n0)to[bend left] (n2)  {} ;
  \draw (n0) -- (n1){};
  \draw (n0) to[bend right] (n1){};
\end{tikzpicture}+
\begin{tikzpicture}
[baseline=(current bounding box.center),decoration={markings, 
    mark= at position 0.6 with {\arrow{Straight Barb[ length=1mm, width=1mm]}}}
] 
  \tikzstyle{every node}=[circle, draw, fill=black,inner sep=0pt, minimum width=3pt]
  \node (n3) at (9,0.7) [label= right: $x$]{};
  \node(n2) at (9,0)  [label= right: $z_1$]{};
   \node(n1) at (9,-0.7)  [label= right: $z_2$]{};
    \node(n0) at (9,-1.4)  [label= right: $z_3$]{};
    \draw [postaction={decorate}](n2) -- (n3)  {} ;
    \draw [postaction={decorate}](n1) -- (n2)  {} ;
     \draw [postaction={decorate}](n0)to[bend left] (n2)  {} ;
  \draw [postaction={decorate}](n0) -- (n1){};
  \draw (n0) to[bend right] (n1){};
\end{tikzpicture}+
\begin{tikzpicture}
[baseline=(current bounding box.center),decoration={markings, 
    mark= at position 0.6 with {\arrow{Straight Barb[ length=1mm, width=1mm]}}}
] 
  \tikzstyle{every node}=[circle, draw, fill=black,inner sep=0pt, minimum width=3pt]
  \node (n3) at (9,0.7) [label= right: $x$]{};
  \node(n2) at (9,0)  [label= right: $z_1$]{};
   \node(n1) at (9,-0.7)  [label= right: $z_2$]{};
    \node(n0) at (9,-1.4)  [label= right: $z_3$]{};
    \draw [postaction={decorate}](n2) -- (n3)  {} ;
    \draw [postaction={decorate}](n1) -- (n2)  {} ;
     \draw [postaction={decorate}](n0)to[bend left] (n2)  {} ;
  \draw [postaction={decorate}](n0) -- (n1){};
  \draw [postaction={decorate}](n0) to[bend right] (n1){};
\end{tikzpicture}+
\begin{tikzpicture}
[baseline=(current bounding box.center),decoration={markings, 
    mark= at position 0.6 with {\arrow{Straight Barb[ length=1mm, width=1mm]}}}
] 
  \tikzstyle{every node}=[circle, draw, fill=black,inner sep=0pt, minimum width=3pt]
  \node (n3) at (9,0.7) [label= right: $x$]{};
  \node(n2) at (9,0)  [label= right: $z_1$]{};
   \node(n1) at (9,-0.7)  [label= right: $z_2$]{};
    \node(n0) at (9,-1.4)  [label= right: $z_3$]{};
    \draw [postaction={decorate}](n2) -- (n3)  {} ;
    \draw [postaction={decorate}](n1) -- (n2)  {} ;
     \draw [](n0)to[bend left] (n2)  {} ;
  \draw [postaction={decorate}](n0) -- (n1){};
  \draw [](n0) to[bend right] (n1){};
\end{tikzpicture}+
\begin{tikzpicture}
[baseline=(current bounding box.center),decoration={markings, 
    mark= at position 0.6 with {\arrow{Straight Barb[ length=1mm, width=1mm]}}}
] 
  \tikzstyle{every node}=[circle, draw, fill=black,inner sep=0pt, minimum width=3pt]
  \node (n3) at (9,0.7) [label= right: $x$]{};
  \node(n2) at (9,0)  [label= right: $z_1$]{};
   \node(n1) at (9,-0.7)  [label= right: $z_2$]{};
    \node(n0) at (9,-1.4)  [label= right: $z_3$]{};
    \draw [postaction={decorate}](n2) -- (n3)  {} ;
    \draw [postaction={decorate}](n1) -- (n2)  {} ;
     \draw [](n0)to[bend left] (n2)  {} ;
  \draw [postaction={decorate}](n0) -- (n1){};
  \draw [postaction={decorate}](n0) to[bend right] (n1){};
\end{tikzpicture}
\right)\Upsilon(z_1, z_2 ,z_3)+\cdots\\
\end{split}
\end{equation}
Note that thus far all our diagrams are labelled---this was necessary in the expansion of the Heisenberg field because the vertex labels kept track of the order of the field operators in the expansion. But in computing the scalar quantity $ \langle \phi^H(x) \rangle$ this is no longer necessary and we now proceed to obtain an expansion in terms of unlabelled diagrams akin to the Feynman diagrams of the continuum.

Consider some labelled diagram $G_n$ with $n$ internal vertices which appears in the expansion of $\langle \phi^H(x)\rangle $ and permute the labels of its internal vertices to produce another diagram $G_n'$. If $G_n'$ does not appears in the expansion of $\langle \phi^H(x)\rangle$, then it must contain a directed edge $z_i\rightarrow z_j$ (which we interpret as $-i\Delta^R_{z_j,z_i}$) for some $i<j$. We conclude that the product $G_n'\Upsilon(z_1,\ldots, z_n)$ vanishes, since $\Upsilon(z_1,\ldots, z_n)$ vanishes when $z_i<z_j$ and $\Delta^R_{z_jz_i}$ vanishes when $z_i\geq z_j$. Using this insight we may write $\langle \phi^H(x)\rangle$ as, 
    \begin{equation}\label{eq in-in correlator D}
    \begin{split}
  &\langle \phi^H(x)\rangle =\sum_{n=0}^{\infty} \sum_{G_n}(-i)^n\sum_{z_1\cdots z_{n}=1}^{x_1-1}G_n \ \Upsilon(z_1\cdots z_n),\\
    \end{split}
    \end{equation}
where now $G_n$ denotes either an allowed labelled diagram or a diagram obtained from one via a permutation of the internal vertex labels. Now, note that the internal vertices are simply dummy indices all of which are summed over the same domain. Therefore, for each isomorphism-class $[G_n]$ choose some representative $G_n$, and for each diagram $G'_n\in[G_n]$ apply the necessary coordinate transformation $(z_1\ldots z_n)\Gto (z_1\ldots z_n)$ to bring $G'_n$ into the labelling of $G_n$. Collecting the identical diagrams and taking into account the action of each coordinate transformation on $\Upsilon$, we obtain, 
 \begin{equation}\label{eq in-in correlator [G] 1field}
    \begin{split}
  &\langle \phi^H(x)\rangle =\sum_{n=0}^{\infty}\sum_{[G_n]} (-i)^n\sum_{z_1\cdots z_{n}=1}^{x_1-1} \frac{G_n}{|Aut(G_n)|} \ \sum_{\pi}\Upsilon(\pi),
    \end{split}
    \end{equation}
where $[G_n]$ denotes an unlablled diagram, $G_n$ denotes a labelled diagram representative of $[G_n]$, $Aut(G_n)$ is the group of automorphisms of $G_n$ which keep the $x$ vertex fixed, $\pi=z_{i_1},\ldots, z_{i_n}$ is a permutation of $z_1,\ldots, z_n$ and,
\begin{equation}
  \Upsilon(z_{i_1},\ldots, z_{i_n})= \begin{cases}
      \frac{1}{l_1!\ldots l_k!} & z_{i_1}\geq z_{i_2}\geq \ldots\geq  z_{i_n}\\
      0 & \text{otherwise,}\\
    \end{cases}    
\end{equation}
where $k$ denotes the number of strings of equal signs in $z_{i_1}\geq z_{i_2}\geq \ldots\geq  z_{i_n}$ and $l_i$ is the length of the $i^{th}$ string of equal signs. Noting that for any set of values for $z_1,\ldots, z_n$, we have $\sum_{\pi}\Upsilon(\pi)=1$, expansion \eqref{eq in-in correlator [G] 1field} simplifies to a sum over unlabelled diagrams $[G_n]$, \begin{equation}\label{eq in-in correlator [G] final 1field}
    \begin{split}
  &\langle \phi^H(x)\rangle=\sum_{n=0}^{\infty}\sum_{[G_n]}[G_n],
    \end{split}
    \end{equation}
where each $[G_n]$ is assigned a value according to the diagrammatic rules summarised below. 

\paragraph{The diagrammatic expansion for $\langle \phi^H(x)\rangle$} The diagrams which contribute to the expansion at order $n$ contain a single external vertex labelled $x$ and $n$ unlabelled internal vertices. The valency of the internal vertices is given by the order of the Hamiltonian in the field and all half-legs are contracted to form legs. Each internal vertex is connected to $x$ by at least one directed path and there are no closed directed cycles. The value of each diagram is obtained by assigning each internal vertex with a dummy index $z$ and following the rules:
\begin{itemize}
\item each internal vertex $z$ gives a sum $i\sum_{z=1}^{x-1}$ and a factor of the coupling and any other constants appearing in the interaction Hamiltonian $\mathcal{H}$,
\item each directed edge $a\rightarrow b$ gives a factor of $-i\Delta^R_{ba}$,
\item  each undirected edge $\begin{tikzpicture}[baseline={(0,-.5ex)}]\node[align=left] at (-0.1,0) {$a$};\draw[black, thick]  (0.15,0)  --  (.85,0);\node[align=left] at (1.,0.0) {$b$};\end{tikzpicture}$ gives a factor of $\Delta^F_{ab}$,
\item multiply by an overall factor coming from the different ways of connecting the half-legs to form the diagram,
\item divide by ${|Aut(G_n)|}$, the number of automorphisms which keep $x$ fixed.
\end{itemize}

For example, expansion \eqref{exp phi3 example} can be rewritten in terms of unlabelled diagrams as,
\begin{equation}\label{exp phi3 unlabelled example}
\langle \phi^H(x)\rangle =
\begin{tikzpicture}
[baseline=(current bounding box.center),decoration={markings, 
    mark= at position 0.6 with {\arrow{Straight Barb[ length=1mm, width=1mm]}}}
] 
  \tikzstyle{every node}=[circle, draw, fill=black,inner sep=0pt, minimum width=3pt]
  \node (n3) at (9,0.7) [label= right: $x$]{};
  \node(n2) at (9,0)  []{};
    \draw [postaction={decorate}](n2) -- (n3)  {} ;
  \draw (n2) to [loop below, scale=1.5, looseness=20,min distance=5mm] (n2){};
\end{tikzpicture} +
\begin{tikzpicture}
[baseline=(current bounding box.center),decoration={markings, 
    mark= at position 0.6 with {\arrow{Straight Barb[ length=1mm, width=1mm]}}}
] 
  \tikzstyle{every node}=[circle, draw, fill=black,inner sep=0pt, minimum width=3pt]
  \node (n3) at (9,0.7) [label= right: $x$]{};
  \node(n2) at (9,0)  []{};
   \node(n1) at (9,-0.7)  []{};
    \node(n0) at (9,-1.4)  []{};
    \draw [postaction={decorate}](n2) -- (n3)  {} ;
    \draw [postaction={decorate}](n1) to[bend left] (n2)  {} ;
    \draw (n1)to[bend right] (n2)  {} ;
    \draw [postaction={decorate}](n0) -- (n1)  {} ;
  \draw (n0) to [loop below, scale=1.5, looseness=5,min distance=5mm] (n0){};
\end{tikzpicture}+\begin{tikzpicture}
[baseline=(current bounding box.center),decoration={markings, 
    mark= at position 0.6 with {\arrow{Straight Barb[ length=1mm, width=1mm]}}}
] 
  \tikzstyle{every node}=[circle, draw, fill=black,inner sep=0pt, minimum width=3pt]
  \node (n3) at (9,0.7) [label= right: $x$]{};
  \node(n2) at (9,0)  []{};
   \node(n1) at (9,-0.7)  []{};
    \node(n0) at (9,-1.4)  []{};
    \draw [postaction={decorate}](n2) -- (n3)  {} ;
    \draw [postaction={decorate}](n1) to[bend left] (n2)  {} ;
    \draw [postaction={decorate}](n1)to[bend right] (n2)  {} ;
    \draw [postaction={decorate}](n0) -- (n1)  {} ;
  \draw (n0) to [loop below, scale=1.5, looseness=5,min distance=5mm] (n0){};
\end{tikzpicture}+
\begin{tikzpicture}
[baseline=(current bounding box.center),decoration={markings, 
    mark= at position 0.6 with {\arrow{Straight Barb[ length=1mm, width=1mm]}}}
] 
  \tikzstyle{every node}=[circle, draw, fill=black,inner sep=0pt, minimum width=3pt]
  \node (n3) at (9,0) [label= right: $x$]{};
  \node(n2) at (9,-0.7)  []{};
   \node(n1) at (8.5,-1.4)  []{};
    \node(n0) at (9.5,-1.4)  []{};
    \draw [postaction={decorate}](n2) -- (n3)  {} ;
    \draw [postaction={decorate}](n1) -- (n2)  {} ;
     \draw [postaction={decorate}](n0)-- (n2)  {} ;
  \draw (n1) to [loop below, scale=1.5, looseness=5,min distance=5mm] (n1){};
  \draw (n0) to [loop below, scale=1.5, looseness=5,min distance=5mm] (n0){};
\end{tikzpicture}+
\begin{tikzpicture}
[baseline=(current bounding box.center),decoration={markings, 
    mark= at position 0.6 with {\arrow{Straight Barb[ length=1mm, width=1mm]}}}
] 
  \tikzstyle{every node}=[circle, draw, fill=black,inner sep=0pt, minimum width=3pt]
  \node (n3) at (9,0.7) [label= right: $x$]{};
  \node(n2) at (9,0)  []{};
   \node(n1) at (9,-0.7)  []{};
    \node(n0) at (9,-1.4)  []{};
    \draw [postaction={decorate}](n2) -- (n3)  {} ;
    \draw [postaction={decorate}](n1) -- (n2)  {} ;
     \draw [postaction={decorate}](n0)to[bend left] (n2)  {} ;
  \draw (n0) -- (n1){};
  \draw (n0) to[bend right] (n1){};
\end{tikzpicture}+
\begin{tikzpicture}
[baseline=(current bounding box.center),decoration={markings, 
    mark= at position 0.6 with {\arrow{Straight Barb[ length=1mm, width=1mm]}}}
] 
  \tikzstyle{every node}=[circle, draw, fill=black,inner sep=0pt, minimum width=3pt]
  \node (n3) at (9,0.7) [label= right: $x$]{};
  \node(n2) at (9,0)  []{};
   \node(n1) at (9,-0.7)  []{};
    \node(n0) at (9,-1.4)  []{};
    \draw [postaction={decorate}](n2) -- (n3)  {} ;
    \draw [postaction={decorate}](n1) -- (n2)  {} ;
     \draw [postaction={decorate}](n0)to[bend left] (n2)  {} ;
  \draw [postaction={decorate}](n0) -- (n1){};
  \draw (n0) to[bend right] (n1){};
\end{tikzpicture}+
\begin{tikzpicture}
[baseline=(current bounding box.center),decoration={markings, 
    mark= at position 0.6 with {\arrow{Straight Barb[ length=1mm, width=1mm]}}}
] 
  \tikzstyle{every node}=[circle, draw, fill=black,inner sep=0pt, minimum width=3pt]
  \node (n3) at (9,0.7) [label= right: $x$]{};
  \node(n2) at (9,0)  []{};
   \node(n1) at (9,-0.7)  []{};
    \node(n0) at (9,-1.4)  []{};
    \draw [postaction={decorate}](n2) -- (n3)  {} ;
    \draw [postaction={decorate}](n1) -- (n2)  {} ;
     \draw [postaction={decorate}](n0)to[bend left] (n2)  {} ;
  \draw [postaction={decorate}](n0) -- (n1){};
  \draw [postaction={decorate}](n0) to[bend right] (n1){};
\end{tikzpicture}+
\begin{tikzpicture}
[baseline=(current bounding box.center),decoration={markings, 
    mark= at position 0.6 with {\arrow{Straight Barb[ length=1mm, width=1mm]}}}
] 
  \tikzstyle{every node}=[circle, draw, fill=black,inner sep=0pt, minimum width=3pt]
  \node (n3) at (9,0.7) [label= right: $x$]{};
  \node(n2) at (9,0)  []{};
   \node(n1) at (9,-0.7)  []{};
    \node(n0) at (9,-1.4)  []{};
    \draw [postaction={decorate}](n2) -- (n3)  {} ;
    \draw [postaction={decorate}](n1) -- (n2)  {} ;
     \draw [](n0)to[bend left] (n2)  {} ;
  \draw [postaction={decorate}](n0) -- (n1){};
  \draw [](n0) to[bend right] (n1){};
\end{tikzpicture}+
\begin{tikzpicture}
[baseline=(current bounding box.center),decoration={markings, 
    mark= at position 0.6 with {\arrow{Straight Barb[ length=1mm, width=1mm]}}}
] 
  \tikzstyle{every node}=[circle, draw, fill=black,inner sep=0pt, minimum width=3pt]
  \node (n3) at (9,0.7) [label= right: $x$]{};
  \node(n2) at (9,0)  []{};
   \node(n1) at (9,-0.7)  []{};
    \node(n0) at (9,-1.4)  []{};
    \draw [postaction={decorate}](n2) -- (n3)  {} ;
    \draw [postaction={decorate}](n1) -- (n2)  {} ;
     \draw [](n0)to[bend left] (n2)  {} ;
  \draw [postaction={decorate}](n0) -- (n1){};
  \draw [postaction={decorate}](n0) to[bend right] (n1){};
\end{tikzpicture}
+\cdots
\end{equation}

\subsection{In-in correlators}\label{subsec: in in correlators}
In Section \ref{subsec: field exp value} we gave diagrammatic rules for computing the expectation value of the Heisenberg field, $\langle \phi^H(x)\rangle$. Here, we generalise these rules to in-in correlators of causally ordered products of local operators. 

Given some causet $\mathcal{C}$ and an integer $1\leq k\leq |\mathcal{C}|$, let $x_1,x_2,\ldots ,x_k$ denote integers satisfying $|\mathcal{C}|\geq x_1>x_2\ldots >x_k\geq 1$. For each $i=1,\ldots, k$ we write $\mathcal{O}^H(x_i)$ to denote a local Heisenberg operator at $x_i$. We allow for the dependence of $\mathcal{O}^H(x_i)$ on $\phi(x_i)$ to be different to the dependence of $\mathcal{O}^H(x_j)$ on $\phi(x_j)$, but we restricts ourselves to operators which are finite-order polynomials in the Heisenberg fields, \textit{e.g.} $\mathcal{O}^H(x_2)=3(\phi^H(x_2))^4$. We seek to compute the in-in correlator,
\begin{equation}
       \langle \mathcal{O}^H(x_1)\ldots \mathcal{O}^H(x_k)\rangle.
    \end{equation}

We proceed in two stages. First, we expand the operator product $\mathcal{O}^H(x_1)\ldots \mathcal{O}^H(x_k)$ as a sum of causally-ordered products of interaction picture fields. Second, we apply Wick's theorem to compute the expectation value of each term in the expansion.

Our strategy for obtaining the expansion of $\mathcal{O}^H(x_1)\ldots \mathcal{O}^H(x_k)$ is to express it as a nested product (cf. Lemma \ref{recursion lemma lemma}),
\begin{equation}\label{def recursive product} \mathcal{O}^H(x_1)\ldots \mathcal{O}^H(x_k)=U_{x_k}^{\dagger}\mathcal{O}_{1\ldots k} U_{x_k},
    \end{equation}
    where $\mathcal{O}_{1\ldots k}$ is defined recursively via,
    \begin{equation}\label{def O1...p}
\mathcal{O}_{1\ldots p} =
    \begin{cases}
      \mathcal{O}(x_{1}) & p=1\\
      U_{x_{p-1},x_p}^{\dagger}\mathcal{O}_{1\ldots p-1}U_{x_{p-1},x_p}\mathcal{O}(x_{p}) & 1<p\leq k.\\
    \end{cases}
\end{equation}
This enables us to obtain the product expansion by the recursive application of the relation (cf. Lemma \ref{lemma BCH appendix}), 
\begin{equation}\label{BCH section equation}\begin{split}
&U_{x,y}^{\dagger}\mathcal{O}U_{x,y} = \sum_{n=0}^{\infty} (-i)^n\sum_{z_1,\ldots z_n=y}^{x-1} \bigg[\ldots\bigg[\bigg[\mathcal{O},\mathcal{H}(z_1)\bigg],\mathcal{H}(z_2)\bigg]\ldots,\mathcal{H}(z_n)\bigg]\Upsilon(z_1,\ldots, z_n) ,\\
\end{split}
\end{equation}
where $\mathcal{O}$ is any (not necessarily local) operator and the $n=0$ term is understood to be equal to $\mathcal{O}$. We leave the result to the Appendix (cf. Corollary \ref{cor operator product expansion}), but note that inside the expectation value the expansion simplifies to (cf. Lemma \ref{lemma correlator form}),    
\begin{equation}\label{eq: product expansion1}
\begin{split}
&\langle \mathcal{O}^H(x_1)\ldots \mathcal{O}^H(x_k)\rangle \\
&\hspace{5mm}=\sum_{n=0}^{\infty}  (-i)^n\sum_{z_1,\ldots,z_n=1}^{x-1}\bigg\langle\bigg[\ldots\bigg[\bigg[\mathcal{O}(x_1)\ldots \mathcal{O}(x_k),\mathcal{H}(z_1)\bigg],\mathcal{H}(z_2)\bigg]\ldots,\mathcal{H}(z_n)\bigg]\bigg\rangle\Upsilon(z_1,\ldots, z_n),\\
\end{split}
\end{equation}
where on the RHS the operators $\mathcal{O}(x_i)$ are in the interaction picture and $\mathcal{H}$ is the interacting Hamiltonian in the interaction picture.

To evaluate the expectation value we apply Wick's theorem diagrammatically. The diagrammatic expansion of the nested correlator in \eqref{eq: product expansion1} is obtained by modifying the rules given in  \ref{diagrammatic expansion for the nested commutator} so that each diagram has $k$ external vertices labelled $x_1,\ldots, x_k$ and each internal vertex is connected to at least one external vertex by at least one directed path. Additionally, the number of half-legs at each external vertex is equal to the number of fields at that point (\textit{e.g.} 3 half-legs at $x_i$ if $\mathcal{O}(x_i)=(\phi^H(x_i))^3$) and there must be no outgoing edges from any of the external vertices. To apply Wick's theorem we connect the remaining half-legs into undirected edges. The resulting diagrams are labelled and we repeat the procedure outlined in Section \ref{subsec: field exp value} to obtained an expansion in terms of unlabelled diagram. 
The rules are obtained from those given in Section \ref{subsec: field exp value} by allowing for multiple external vertices and accounting for any additional factors they may carry. For completeness, the rules are presented below.

\paragraph{The diagrammatic expansion for $\langle \mathcal{O}^H(x_1)\ldots \mathcal{O}^H(x_k)\rangle $} The diagrams which contribute to the expansion at order $n$ contain $k$ external vertices labelled $x_1,\ldots, x_k$ and $n$ unlabelled internal vertices. The valency of the internal (external) vertices is given by the order of the Hamiltonian (local operator $\mathcal{O}$) in the field and all half-legs are contracted to form legs. Each internal vertex is connected to at least one external vertex by at least one directed path. There are no closed directed cycles and no edges directed outwards from an external vertex. The value of each diagram is obtained by assigning each internal vertex with a dummy index $z$ and following the rules:
\begin{itemize}
\item each internal vertex $z$ gives a sum $i\sum_{z=1}^{x-1}$ and a factor of the coupling and any other constants appearing in the interaction Hamiltonian $\mathcal{H}$,
\item each external vertex $x_i$ gives any constant factors appearing in $\mathcal{O}(x_i)$,
\item each directed edge $a\rightarrow b$ gives a factor of $-i\Delta^R_{ba}$,
\item  each undirected edge $\begin{tikzpicture}[baseline={(0,-.5ex)}]\node[align=left] at (-0.1,0) {$a$};\draw[black, thick]  (0.15,0)  --  (.85,0);\node[align=left] at (1.,0.0) {$b$};\end{tikzpicture}$ gives a factor of $\Delta^F_{ab}$,
\item multiply by an overall factor coming from the different ways of connecting the half-legs to form the diagram,
\item divide by ${|Aut(G_n)|}$, the number of automorphisms which keep $x_1,\ldots, x_k$ fixed.
\end{itemize}

See \cite{albertini:2021, Jubb:2023mlv} for an explicit computation of the two point function in $\phi^4$ theory. In $\phi^3$, the connected part of $\langle \phi^H(x_1)\phi^H(x_2) \phi^H(x_3)\rangle$ to first order in the coupling is given by,
\begin{equation}\label{3ptfn unlabelled example}\begin{split}
& 
\begin{tikzpicture}
[baseline=(current bounding box.center),decoration={markings, 
    mark= at position 0.6 with {\arrow{Straight Barb[ length=1mm, width=1mm]}}}
] 
  \tikzstyle{every node}=[circle, draw, fill=black,inner sep=0pt, minimum width=3pt]
  \node (n3) at (9,0.7) [label= above: $x_2$]{};
  \node (n2) at (8.5,0.7) [label= left: $x_1$]{};
  \node (n1) at (9.5,0.7) [label= right: $x_3$]{};
  \node(n0) at (9,0)  []{};
    \draw [postaction={decorate}](n0) -- (n3)  {} ;
    \draw [postaction={decorate}](n0) -- (n2)  {} ;
    \draw [postaction={decorate}](n0) -- (n1)  {} ; 
\end{tikzpicture} 
+
\begin{tikzpicture}
[baseline=(current bounding box.center),decoration={markings, 
    mark= at position 0.6 with {\arrow{Straight Barb[ length=1mm, width=1mm]}}}
] 
  \tikzstyle{every node}=[circle, draw, fill=black,inner sep=0pt, minimum width=3pt]
  \node (n3) at (9,0.7) [label= above: $x_2$]{};
  \node (n2) at (8.5,0.7) [label= left: $x_1$]{};
  \node (n1) at (9.5,0.7) [label= right: $x_3$]{};
  \node(n0) at (9,0)  []{};
    \draw [](n0) -- (n3)  {} ;
    \draw [postaction={decorate}](n0) -- (n2)  {} ;
    \draw [postaction={decorate}](n0) -- (n1)  {} ; 
\end{tikzpicture} 
+
\begin{tikzpicture}
[baseline=(current bounding box.center),decoration={markings, 
    mark= at position 0.6 with {\arrow{Straight Barb[ length=1mm, width=1mm]}}}
] 
  \tikzstyle{every node}=[circle, draw, fill=black,inner sep=0pt, minimum width=3pt]
  \node (n3) at (9,0.7) [label= above: $x_2$]{};
  \node (n2) at (8.5,0.7) [label= left: $x_1$]{};
  \node (n1) at (9.5,0.7) [label= right: $x_3$]{};
  \node(n0) at (9,0)  []{};
    \draw [postaction={decorate}](n0) -- (n3)  {} ;
    \draw [](n0) -- (n2)  {} ;
    \draw [postaction={decorate}](n0) -- (n1)  {} ;  
\end{tikzpicture}\\
  & +
 \begin{tikzpicture}
[baseline=(current bounding box.center),decoration={markings, 
    mark= at position 0.6 with {\arrow{Straight Barb[ length=1mm, width=1mm]}}}
] 
  \tikzstyle{every node}=[circle, draw, fill=black,inner sep=0pt, minimum width=3pt]
  \node (n3) at (9,0.7) [label= above: $x_2$]{};
  \node (n2) at (8.5,0.7) [label= left: $x_1$]{};
  \node (n1) at (9.5,0.7) [label= right: $x_3$]{};
  \node(n0) at (9,0)  []{};
    \draw [postaction={decorate}](n0) -- (n3)  {} ;
    \draw [postaction={decorate}](n0) -- (n2)  {} ;
    \draw [](n0) -- (n1)  {} ;
\end{tikzpicture}
+
\begin{tikzpicture}
[baseline=(current bounding box.center),decoration={markings, 
    mark= at position 0.6 with {\arrow{Straight Barb[ length=1mm, width=1mm]}}}
] 
  \tikzstyle{every node}=[circle, draw, fill=black,inner sep=0pt, minimum width=3pt]
  \node (n3) at (9,0.7) [label= above: $x_2$]{};
  \node (n2) at (8.5,0.7) [label= left: $x_1$]{};
  \node (n1) at (9.5,0.7) [label= right: $x_3$]{};
  \node(n0) at (9,0)  []{};
    \draw [](n0) -- (n3)  {} ;
    \draw [](n0) -- (n2)  {} ;
    \draw [postaction={decorate}](n0) -- (n1)  {} ;  
\end{tikzpicture} 
+
\begin{tikzpicture}
[baseline=(current bounding box.center),decoration={markings, 
    mark= at position 0.6 with {\arrow{Straight Barb[ length=1mm, width=1mm]}}}
] 
  \tikzstyle{every node}=[circle, draw, fill=black,inner sep=0pt, minimum width=3pt]
  \node (n3) at (9,0.7) [label= above: $x_2$]{};
  \node (n2) at (8.5,0.7) [label= left: $x_1$]{};
  \node (n1) at (9.5,0.7) [label= right: $x_3$]{};
  \node(n0) at (9,0)  []{};
    \draw [](n0) -- (n3)  {} ;
    \draw [postaction={decorate}](n0) -- (n2)  {} ;
    \draw [](n0) -- (n1)  {} ;
\end{tikzpicture} 
+
\begin{tikzpicture}
[baseline=(current bounding box.center),decoration={markings, 
    mark= at position 0.6 with {\arrow{Straight Barb[ length=1mm, width=1mm]}}}
] 
  \tikzstyle{every node}=[circle, draw, fill=black,inner sep=0pt, minimum width=3pt]
  \node (n3) at (9,0.7) [label= above: $x_2$]{};
  \node (n2) at (8.5,0.7) [label= left: $x_1$]{};
  \node (n1) at (9.5,0.7) [label= right: $x_3$]{};
  \node(n0) at (9,0)  []{};
    \draw [postaction={decorate}](n0) -- (n3)  {} ;
    \draw [](n0) -- (n2)  {} ;
    \draw [](n0) -- (n1)  {} ;
\end{tikzpicture} \ \ .
\end{split}
\end{equation}

\section{Path integrals}\label{sec: path int}
The interacting quantum field theory on a causal set described above was first proposed in path integral or ``histories'' form as a decoherence functional by Rafael Sorkin in \cite{Sorkin:2011pn}. The decoherence functional, or double path integral of Schwinger-Keldysh form, is the basis for an alternative foundation for quantum theory in which histories and events, rather than operators and states, are fundamental (see for example \cite{Sorkin:1994dt}). 
In this section we will not dwell on this foundational aspect of the path integral approach but instead briefly show how the interacting decoherence functional proposed by Sorkin encapsulates, mathematically, the time-ordered correlation functions in the in-in formalism calculated in Section \ref{sec: in in}. More details of the translation between the operator formalism and the path-integral formalism are given in \cite{albertini:2021} and expanded on in \cite{Jubb:2023mlv}. 
We use the decoherence functional to write down a generating functional for the in-in correlators of Section \ref{subsec: in in correlators}. We also give the generating functional for in-out correlators.
\subsection{The decoherence functional of the free theory}
Consider a causal set $\calC$ with $N$ elements. The space of histories of a real scalar quantum field theory on $\calC$ is the space of real vectors $\RR^N$. We denote a vector by $\xi$ with components $\xi_x$, $x \in \calC$. The decoherence functional for the free theory in a particular state is given by \cite{Sorkin:2011pn},
 \begin{align}
    D_0(\xi,\bar{\xi}) & = \left\langle \delta(\phi_1-\bar{\xi}_1)\delta(\phi_2-\bar{\xi}_2)...\delta(\phi_N-\bar{\xi}_N)\delta(\phi_N-\xi_N)...\delta(\phi_2-\xi_2)\delta(\phi_1-\xi_1)\right\rangle \,,\label{eq:dcf0}
\end{align} 
where we write $\phi_x$ as a shorthand for the free field operator $\phi(x)$ and $\left\langle \cdot \right\rangle$ denotes expectation value in a Gaussian state, $\rho$, which may be the SJ state or another state:
 \begin{align}
  \langle \hat{O} \rangle = \rho\left(\hat{O}\right) = \textrm{Tr}\left[\hat\rho \hat{O} \right] \,.
\end{align} 

\noindent The decoherence functional should have a label indicating its dependence on the state $\rho$ but we omit it for convenience.\\

\noindent Note that:
\begin{itemize}
\item the decoherence functional is normalised
 \begin{equation}\label{eq:normalisation}
    \int d^N \xi\, d^N \bar{\xi}\,\, D_0(\xi,\bar{\xi}) = 1,
\end{equation} 
\item the ordering of the operators in the expectation value (\ref{eq:dcf0}) uses the label order of $\calC$,  
\item  
since operators at spacelike points commute, we have
 \begin{align}
    D_0(\xi,\bar{\xi}) & = \left\langle \overline{C} [\delta(\phi_1-\bar{\xi}_1)\delta(\phi_2-\bar{\xi}_2)... \delta(\phi_N-\bar{\xi}_N)]\,\,C[\delta(\phi_N-\xi_N)...\delta(\phi_2-\xi_2)\delta(\phi_1-\xi_1)]\right\rangle \,,\label{eq:dcf0}
\end{align} 
where $C$ and $\overline{C}$ are the causal and anti-causal ordering operators defined in Section \ref{sec label order}, so the $\xi$ delta-functions are causally ordered and the $\bar{\xi}$ delta-functions are anti-causally ordered, 
\item the delta-functions ensure that the decoherence functional vanishes unless the field values $\xi_x = \bar{\xi}_x$ on every maximal element $x$ of $\calC$,
\item the decoherence functional (\ref{eq:dcf0}) is evaluated in \cite{Sorkin:2011pn}.
\end{itemize}

Let $F(\xi)$ and $G(\xi)$ be real functions on $\RR^N$. If we integrate $F(\xi)G(\bar{\xi})$ over all \textit{pairs} of field configurations, $(\xi,\bar{\xi})\in \mathbb{R}^N\times \mathbb{R}^N$ against a \textit{measure} which equals the decoherence functional, the delta-functions in  $D_0(\xi,\bar{\xi})$ act to causally/anti-causally order the corresponding functions of the field operators:
\begin{equation}\label{eq: df master equation}
\int_{\mathbb{R}^{2N}}d^N\bar{\xi}d^N\xi \, D_0(\xi,\bar{\xi}) F(\xi) G(\bar{\xi}) \,  = \left\langle \overline{C}\left[ G(\phi) \right] C\left[ F(\phi) \right]\right\rangle.
\end{equation}
The simplest examples are the $2$-point correlators and we have,
\begin{align}
\left\langle C[ \phi_x  \phi_y ] \right\rangle &=\int_{\mathbb{R}^{2N}} d^N\bar{\xi}d^N\xi \,\label{eq:causal 2pt func dcf0} D_0(\xi,\bar{\xi}) \, \xi_x \xi_y,\\
\left\langle \overline{C}[ \phi_x  \phi_y ] \right\rangle &=\int_{\mathbb{R}^{2N}} d^N\bar{\xi}d^N\xi \, D_0(\xi,\bar{\xi}) \, \bar\xi_x \bar\xi_y,\\
W_{xy} = \left\langle \phi_x  \phi_y  \right\rangle &=\int_{\mathbb{R}^{2N}} d^N\bar{\xi}d^N\xi \, D_0(\xi,\bar{\xi}) \, \bar\xi_x \xi_y\,.
\end{align}
Slightly more generally, when the functions are monomials in the field components we have
\begin{equation}
\int_{\mathbb{R}^{2N}}d^N\bar{\xi}d^N\xi \, D_0(\xi,\bar{\xi}) \xi_{x_1}\dots \xi_{x_l}\,\,\bar\xi_{y_1}\dots\bar\xi_{y_m} \,  = \left\langle \overline{C}\left[\phi_{y_1}\dots\phi_{y_m} \right] C\left[ \phi_{x_1}\dots \phi_{x_l} \right] \right\rangle\,.
\end{equation}
The causal ordering results from the fact that in the integral, the factor $\xi_x$, say, is just a real variable and can be moved anywhere. When $\xi_x$ is moved next to the corresponding operator factor 
$\delta(\phi_x - \xi_x)$ in the \dcf, $\xi_x$ becomes the operator $\phi_x$ which, as an operator, is now in the causally ordered position of $\delta(\phi_x - \xi_x)$ in the operator product.

\subsection{The decoherence functional of the interacting theory}
In the interacting theory with a $\phi^4$ interaction, Sorkin proposed the interacting decoherence functional,
 \begin{align}
    D_g(\xi,\bar{\xi}) & = D_0(\xi,\bar\xi) \, e^{- i ( \xi^4 - \ {\bar{\xi}}^4) \cdot {{g}}  },
\end{align} 
where,
\begin{align} 
{g}  = (g_1, g_2, g_3, \dots g_N),
\end{align}
is a vector of $N$ coupling constants, one for each causet element and,
\begin{align} 
\calV_{int}(\xi) = { \xi}^4\cdot{g}  : = \sum_{x = 1}^N (\xi_x)^4 g_x, 
\end{align}
is the self-interaction potential. 

The interacting \dcf\  can be generalised to the case of any real polynomial self-interaction. In other words, the $\phi^4$ interaction potential can be replaced by,
    \begin{align} 
 \sum_{x = 1}^N (\xi_x)^4 g_x \longrightarrow
 \calV_{int}(\xi) = \sum_{x = 1}^N \calP_x(\xi_x)\,,
\end{align}
where each local $\calP_x$ is a real polynomial, that may vary from element to element. 
Then the interacting \dcf\ is, 
 \begin{align}\label{intdcf}
    D_g(\xi,\bar{\xi}) & = D_0(\xi,\bar\xi) \, e^{- i \calV_{int}(\xi) +i \calV_{int}(\bar\xi)   }.
\end{align}

Note that:
\begin{itemize}
\item the interacting decoherence functional is normalised,   \textit{i.e.}, the integral of (\ref{intdcf}) over $\xi$ and $\bar\xi$ equals 1 \cite{albertini:2021,Jubb:2023mlv},
\item the varying local polynomial allows for an interaction that is zero outside some interaction region of the causal set, 
\item another generalisation is to $M$  scalar fields, $\phi^{(1)}, \phi^{(2)}, \dots \phi^{(M)}$, with an interaction potential at each element that is a polynomial in the $M$ fields at that element, 
\item in making the correspondence with the continuum, if the causal set is well approximated by a spacetime region $\mathcal{M}$ of dimension $d$ and spacetime volume $V$, then we have for example in the $\phi^4$ case with constant coupling $g_x=g$ for all $x$,
      \begin{align} 
g\sum_{x = 1}^N \,  \,(\xi_x)^4 \longleftrightarrow
\lambda\int_{\mathcal{M}} d^dx \sqrt{-g}\,  \,\xi(x)^4
\end{align}
and so 
 $g  \longleftrightarrow dV \lambda $
where $dV = V/N = l^{d}$, $l$ is the discreteness scale and $\lambda$ is the coupling constant in the continuum.
\end{itemize}

\noindent The interacting causally ordered $n$-point function is given in path integral form as before, with the free \dcf-measure replaced by the interacting \dcf\ (\ref{intdcf}),
\begin{equation} \label{intncorr}
    \left\langle C\left[ \phi^H_{x_1} \dots \phi^H_{x_n}\right]\right\rangle = \int_{\mathbb{R}^{2N}}d^N\bar{\xi}\,d^N \xi \, D_g(\xi,\bar{\xi}) \,  \xi_{x_1}\dots\xi_{x_n}\,.
\end{equation}
It can be shown \cite{albertini:2021,Jubb:2023mlv} that this equals the interacting in-in causally ordered $n$-point function of the Section \ref{sec: in in} in a given state $\rho_{in}$. That is, the  Heisenberg field operators are given by (\ref{Heisenberg}) where $\calH_x = \calP_x(\phi_x)$ for each $x \in \calC$ and $\phi_x$ are the interaction picture fields. $\rho_{in}$ depends on the Gaussian state $\rho$ that defines the free \dcf\ (\ref{eq:dcf0}) in the following way. Since the algebra of observables is generated by the interaction picture field operators (see Section \ref{sec:algebras}) the in-state $\rho_{in}$ is defined by its value on every product of pairs of interaction picture field operators,  which values are extended to all monomials of interaction picture operators (and thence by linearity to all polynomials) by Wick's theorem. So the state $\rho_{in}$ is fully specified by the Wightman function of $\rho$ in the free theory:
\begin{align} \label{definstate}
\rho_{in}\left(\phi_x\phi_y\right)= W_{xy}\,, 
\end{align}
(plus the one point functions if they are nonzero in $\rho$).

In the case that $\rho$ is a pure state $|\Psi\rangle$, then $\rho_{in}$ is also a pure state $|\Psi_{in}\rangle$. 

\subsection{Generating functionals}
The in-in generating functional is given in terms of the \dcf: 
\begin{equation}
    \mathcal{Z}^{in-in}[J,\bar J]=\int d^N\xi\, d^N\bar\xi  \, D_g(\xi,\bar{\xi})  \, e^{-i J.\xi} \, e^{i \bar J.\bar \xi},
\end{equation}
where $J$ and $\bar J$ are two independent sources. For example, the 
in-in causally ordered $2$-point correlator is given by 
\begin{equation}
 i \frac{\partial}{\partial J_x}i\frac{\partial}{\partial J_y} \mathcal{Z}^{in-in}[J,\bar J]\Big| _{J=0,\bar J=0} 
  =   \left\langle C\left[\phi^H_x\phi^H_y\right] \right\rangle_{\rho_{in}}\,,
\end{equation}
with similar expressions with more derivatives for the causally ordered in-in $n$-point correlators. Derivatives with respect to $\bar{J}$ similarly result in anti-causally ordered products of field operators. 

The generating functional for the causally ordered in-out correlators is not given in terms of the interacting \dcf\ but by a closely related expression,
\begin{equation}
     \mathcal{Z}^{in-out}[J]=\frac{\int d^N\xi d^N\bar\xi  \, D_0(\xi,\bar\xi)\,  e^{-i  \calV_{int}(\xi) } \, e^{-i J \cdot \xi}}{\int d^N\xi \,d^N \bar\xi  \, D_0(\xi,\bar\xi)\,  e^{-i\calV_{int}(\xi)}}\,,
\end{equation}
where the case of the $\phi^4$ interaction is shown (again, this can be generalised to any polynomial local interaction).  
Now we have,
\begin{equation}\label{inout2point}    i\frac{\partial}{\partial J_x}i\frac{\partial}{\partial J_y} \mathcal{Z}^{in-out}[J]\Big| _{J=0} = \frac{\langle \hat{S}\, C[\phi^H_x\phi^H_y  ]\rangle}{\langle \hat{S}\rangle}  \,,
\end{equation}
where the S-operator is given by, \begin{equation}\label{eq: S operator}\hat{S}= C[\prod_{z\in \mathcal{C}} e^{-i \calH(z)}]\end{equation} and the $C$ operator applied to Heisenberg field acts as, \begin{equation} C[\phi^H(x)\phi^H(y)]=
    \begin{cases}
      \phi^H(x)\phi^H(y) & \text{if $x\succ y$}\\
       \phi^H(y)\phi^H(x) & \text{otherwise}.\\
    \end{cases}\end{equation} To see this, note that the derivatives have no effect on the denominator of the generating functional so the denominator of \eqref{inout2point} follows directly from \eqref{eq: df master equation}. The numerator, after taking derivatives and setting $J=0$, becomes 
$\langle C[\prod_{z}e^{-i  \mathcal{H}(z)} \,\phi_x\phi_y]\rangle $. When $x>y$, the causally-ordered product can be expanded as, \begin{equation}
    \langle e^{-iH(N)}\ldots e^{-iH(x)}U_x\phi^H_xU_x^{\dagger}e^{-iH(x-1)}\ldots e^{-iH(y)}U_y\phi^H_yU_y^{\dagger}e^{-iH(y-1)}\ldots e^{-iH(1)}\rangle=\langle \hat{S}\phi_x^H\phi_y^H\rangle,
\end{equation}
where on the LHS we used relation \eqref{Heisenberg U} between the Heisenberg and interaction picture fields and on the RHS we simplified the products of exponentials. Similarly, when $x<y$
we obtain $\langle \hat{S}\phi_y^H\phi_x^H\rangle$, so overall the numerator is given by $\langle \hat{S}C[\phi_x^H\phi_y^H]\rangle$.

In contrast to the in-in correlators, the expansion of the in-out correlators as a power series in the coupling(s) does not terminate at a finite order even if the causal set itself is finite. One can see this by expanding \eqref{inout2point} diagrammatically: the diagrams are identical to those of the continuum, with ${\langle \hat{S}\rangle}$ given by the sum of vacuum bubble diagrams. 
\section{Scattering amplitudes}\label{sec: scattering}
In this Section, we propose a definition for the S-matrix on a causal set. In curved spacetime, particles can be produced by a non-static metric and one expects this also in QFT on a causal set that is a sprinkling into a nonstatic spacetime. On the causal set, there are many continuum structures and concepts that are used in continuum QFT that are missing. For example, there is no analogue of a Cauchy surface in a causal set. We progress by adapting the concept of ``scattering'' to the technology we do have on the causal set: the SJ vacuum with its corresponding particle states, noting that in the continuum, in a spacetime with a timelike Killing vector, one can show formally that the SJ vacuum equals the usual canonical vacuum defined by the positive frequency modes \cite{Afshordi:2012jf}. To get as close as possible to the usual set-up for scattering calculations, we assume that the interaction region is confined to a region between non-interacting ``past'' and ``future'' regions. This enables the definition of ``asymptotic'' states formally associated to these non-interacting past and future regions, which states then provide a proposed definition for the causal set S-matrix.

To define the past and future regions, we will make use of the following terminology. A subcauset $\mathcal{S}\subseteq \mathcal{C}$ is called a \textit{stem} (or \textit{down-set}) if $x\in \mathcal{S}$ implies that $y\in \mathcal{S}$ for all $y\prec x$. A subcauset $\mathcal{S}\subseteq \mathcal{C}$ is called a \textit{total stem} if it is a stem and for all $y\not\in \mathcal{S}$ there exists some $x\in \mathcal{S}$ such that $x\prec y$.  A subcauset $\mathcal{S}\subseteq \mathcal{C}$ is called an \textit{up-set} if $x\in \mathcal{S}$ implies that $y\in \mathcal{S}$ for all $y\succ x$. A subcauset $\mathcal{S}\subseteq \mathcal{C}$ is called a \textit{total up-set} if it is an up-set and for all $y\not\in \mathcal{S}$ there exists some $x\in \mathcal{S}$ such that $x\succ y$.  

Consider a finite causet $\mathcal{C}$ with $N>1$ elements. Let $\mathcal{P}$ and $\mathcal{F}$ denote a total down-set and a total up-set in $\mathcal{C}$, respectively, with $\mathcal{P}\cap\mathcal{F}=\emptyset$. We will refer to $\mathcal{P}$ and $\mathcal{F}$ as the past and future regions in $\mathcal{C}$ and require that there are no interactions in these regions, \textit{i.e.}, the interaction region is given by $\mathcal{C}\setminus \{\mathcal{P}\cup\mathcal{F}\}$. Recall that we require that $\mathcal{C}$ is naturally labelled, \textit{i.e.} that if $x\prec y$ then $x<y$ (cf. Section \ref{subsec: causets}). Therefore, regardless of which natural labelling one chooses, the element labelled 1 is always contained in the past region $\mathcal{P}$ and the element labelled $N$ is always contained in the future region $\mathcal{F}$. This will be the key to defining our ``asymptotic'' states.

To define particle states in the interacting theory, note that we can obtain a mode expansion of the Heisenberg field by substituting the free field mode expansion \eqref{free field mode expansion} into definition \eqref{Heisenberg U} of the Heisenberg field,
\begin{equation}
\phi^H(x)=\sum_{\lambda>0} \sqrt{\lambda}\bigg( v^{\lambda}_x b_{\lambda}(x)+\bar{v}^{\lambda}_x b_{\lambda}^{\dagger}(x)\bigg),\end{equation}
where we define $b_{\lambda}(x)=U_x^\dagger a_{\lambda} U_x$.

At $x=1$, we recover the free theory ladder operators, $b_{\lambda}(1)=a_{\lambda}$, and the vacuum, denoted by $|0;1\rangle$ and defined by the requirement that $b_{\lambda}(1)|0;1\rangle=0$ for all $\lambda$, is the SJ vacuum (cf. Section \ref{sec:freetheory}).

At $x=N$, the vacuum, denoted by $|0;N\rangle$, is given by,
\begin{equation}\begin{split}
     b_{\lambda}(N) |0;N\rangle=U_N^{\dagger} a_{\lambda}&U_N|0;x\rangle=0
        \implies U_N|0;x\rangle \propto |0\rangle 
         \implies|0;N\rangle \propto U_N^{\dagger}|0\rangle, \\
\end{split}
\end{equation}
where $\propto$ denotes equality up to an overall phase.

Finally, we define the number operator at $x=1$ and $x=N$ as, \begin{equation}
    \mathcal{N}(x)=\sum_{\lambda} b^{\dagger}_{\lambda}(x)b_{\lambda}(x).
\end{equation}
We can now define an in-state as an eigenstate of $\mathcal{N}(1)$ and an out-state as an eigenstate of $\mathcal{N}(N)$. One can verify that the in-states are given by applications of $b^{\dagger}_{\lambda}(1)$ on $|0;1\rangle$, while the out-states are given by applications of $b^{\dagger}_{\lambda}(N)$ on $|0;N\rangle$. We choose our normalisation to be such that a 1-particle state is given by $|\lambda;x\rangle~=~\frac{1}{\sqrt{\lambda}}b^{\dagger}_{\lambda}(x)|0;x\rangle$ so that $\langle 0;x|\phi(x)|\lambda;x\rangle~=~v_x^{\lambda}$ for $x=1,N$.

(Note that while it may seem that the notion of the number operator and the associated particle states could be extended to any $1\leq x\leq N$ this is problematic because the ``state at $x$'' which one obtains in this way is generally label-dependent unless $x=1,N$. Replacing $U_x$ by the covariant $V_x$ (cf. Section \ref{sec: interacting fields}) results in label-independent states but does not allow for the notion of asymptotic regions, requiring instead that a state be defined at a point.)

Scattering amplitudes are given by the overlap of an in- and an out-state. Taking 2-to-2 scattering as an example, the associated amplitude is,
\begin{equation}
    \begin{split}\label{eq:scatteringamp}
       _{out}\langle \lambda_3,\lambda_4 |\lambda_1,\lambda_2\rangle_{in}= \langle \lambda_3,\lambda_4 ; N|\lambda_1,\lambda_2; 1\rangle=\langle \lambda_3, \lambda_4|  \hat{S} |\lambda_1,\lambda_2\rangle,
    \end{split}
\end{equation}
where on the RHS, $|\lambda,\lambda'\rangle$ are the particle states of the free theory and $\hat{S}$ is the S-operator given in \eqref{eq: S operator} which reduces to,
\begin{equation}
    \hat{S}=C[\prod_ze^{-i\mathcal{H}(z)}],
\end{equation}
where the product is over points $z$ in the interaction region.

To evaluate S-matrix elements, one follows the familiar continuum prescription: expand $U_x$ order by order in the interaction coupling and apply Wick's theorem. A contraction of a pair of fields gives rise to a Feynman propagator (represented by an internal leg)  while a contraction of a field with a 1-particle state gives a mode function,
\begin{equation}
    \contraction{}{\phi}{(x)}{|\lambda\rangle}
\phi(x)| \lambda\rangle=v_x^{\lambda}|0\rangle,
\end{equation}
and is represented by an external leg. From (\ref{eq:scatteringamp}) we see that in the absence of interactions, the amplitude trivializes and it does not capture any particle production due to a non-static metric. We will investigate this in future work.

\section{Conclusion}\label{sec: conclusion}
\paragraph{Summary} In the causal set approach to quantum gravity, spacetime is fundamentally discrete at the Planck scale and takes the form of a causal set. Establishing a framework for quantum field theory on causal sets is a valuable exercise: for describing matter on a causal set spacetime, for making phenomenological predictions under the assumption of spacetime discreteness and for exploring new avenues for regularising the UV divergences of the continuum through numerical methods. The main result of this work extends the body of work on the in-in formalism on a causal set by proving that the work of \cite{albertini:2021,Jubb:2023mlv} on 2-point functions in $\phi^4$ theory generalises to all causally ordered in-in correlators of local operators in scalar theories. This approach complements the construction in \cite{Dable-Heath:2019sej} of algebraic quantum field theory on a causal set and the associated diagrammatic expansions of \cite{Hawkins:2016ddq}. We leave understanding the relationship between the two approaches to future work.

Our approach adapted the continuum construction of \cite{Dickinson:2013lsa} to a causal set background and recovered the same diagrammatic expansion. However, there are some differences between the continuum and discrete which stem from the difference in the interpretation of the diagrams. The fact that in the discrete the continuum integrals over the interaction vertices are replaced by sums mean that several interaction vertices can live at the same spacetime point. And since the diagrams can be interpreted as decorated sub-causal sets, when the interaction region is finite there are only finitely many of them and the series terminates at a finite order.

\textbf{Future directions}
By sprinkling causal sets at large density $\rho$ into cosmological spacetimes, our framework can be applied to the computation of cosmological observables. Each diagram becomes a random variable and understanding how their behaviour depends on $\rho$ is of interest. In this scenario, the coupling constant should be scaled by $1/\rho$ (on dimensional grounds), \textit{e.g.} $\mathcal{H}\sim \frac{g}{4!\rho}\phi^4$. There are many avenues one can explore here, for example computing the mean of a particular diagram over a sample of sprinklings and the fluctuations around this average. A question of particular interest for comparison with the continuum is whether the mean tends to a (finite) limit when $\rho\rightarrow\infty$.

The similarities between the discrete in-in formalism and the continuum formalism of \cite{Dickinson:2013lsa} make comparisons between the discrete and the continuum possible. For example, one can show that in 1+1-dimensional Minkowski space with an interaction region confined to a finite causal diamond, the contribution from a diagram containing only retarded propagators (and no Feynman propagators) is equal to the $\rho\rightarrow\infty$ limit of the ensemble average of the same diagram in the discrete. The key that makes this calculation possible is that in 1+1 dimensions the retarded propagator is constant. Investigating whether this correspondence persists for diagrams which contain Feynman propagators/in higher dimensions is another direction for future research.

\acknowledgments 
The authors thank Calvin Chen, Eli Hawkins, Ian Jubb, Kasia Rejzner, Rafael Sorkin, Andrew Tolley and Yasaman Yazdi for discussions on this work. EA is supported by the STFC Consolidated Grant ST/W507519/1. FD acknowledges the support of the Leverhulme/Royal Society interdisciplinary APEX grant APX/R1/180098. FD is supported in part by STFC grants ST/P000762/1 and ST/W006537/1. Research at Perimeter Institute is supported by the Government of
Canada through Industry Canada and by the Province of Ontario through the Ministry
of Economic Development and Innovation. AN is funded by the President’s PhD Scholarship from Imperial College London and the Canada First Research Excellence Fund through the Arthur B. McDonald Canadian Astroparticle Physics Research Institute. AN also appreciates the hospitality received as a visitor at the Perimeter Institute for Theoretical Physics during the completion of this work. SZ is supported by STFC grant ST/W006537/1 and STFC Consolidated Grant ST/X000575/1.

\appendix
\section{Derivation of the diagrammatic expansions}\label{in-in derivation appendix}
In the following, we denote an ordered list by bold letters, \textit{e.g.} $\textbf{a}=a_1\ldots a_n$ is an ordered list of length $n$. We only consider lists whose entries take value in the natural numbers.

\begin{definition}[The function $\Upsilon$]\label{Upsilon definition appendix} Given an ordered list $\textbf{a}=a_1a_2\ldots a_n$, we define,
\begin{equation}
  \Upsilon(\textbf{a})= \begin{cases}
      \frac{1}{l_1!\ldots l_k!} & a_1\geq a_2 \geq \ldots \geq z_n\\
      0 & \text{otherwise,}\\
    \end{cases}    
\end{equation}
where $k$ denotes the number of strings of equal signs in $a_1\geq a_2\geq \ldots\geq  a_n$ and $l_i$ is the length of the $i^{th}$ string of equal signs.
\end{definition}
For example, if $\textbf{a}=1,2,2,4,5,5,5$ and $\textbf{b}=1,4,2,4$ then $\Upsilon(\textbf{a})=\frac{1}{12}$ and $\Upsilon(\textbf{b})=0$.
\begin{definition}[Shuffle of ordered lists]
    Given a pair of ordered lists, $\textbf{a}$ and $\textbf{b}$, their shuffle, denoted by $\textbf{a}\shuffle\textbf{b}$ is the set of ordered lists which can be constructed by putting together all of the entries of $\textbf{a}$ and $\textbf{b}$ into a single list $\textbf{c}$ such that the order of elements in $\textbf{c}$  is compatible with the order of elements in $\textbf{a}$ and $\textbf{b}$.
\end{definition}
For example, if $\textbf{a}=a_1 a_2$ and $\textbf{b}=b_1$ then $\textbf{a}\shuffle\textbf{b}=\{ a_1a_2b_1, a_1b_1a_2, b_1a_1a_2\}$.

\begin{definition}[Permutation of an ordered list]
   Given an ordered list $\textbf{a}=a_1a_2\ldots a_n$, a permutation $\pi$ of $\textbf{a}$ is an ordered list of length $n$ such that there exists a bijection $f:\textbf{a}\rightarrow\pi$ with $f(a_i)=a_i$ for all $i$.    
    \end{definition}
We will make use of the following properties of the $\Upsilon$ function. 
For any list $\textbf{a}$,
\begin{equation}
    \sum_{\pi}\Upsilon(\pi)=1,
\end{equation}
where the sum is over all permutations $\pi$ of $\textbf{a}$. For any pair of lists, $\textbf{a}$ and $\textbf{b}$,
\begin{equation}\label{eq upsilon identity}
\Upsilon(\textbf{a})\Upsilon(\textbf{b})=\sum_{\textbf{c}\in \textbf{a}\shuffle \textbf{b}}\Upsilon(\textbf{c}),
\end{equation}
where $\shuffle$ denotes the shuffle product.

\begin{lemma}\label{lemma Uxy form}
    \begin{equation}
        U_{x,y}=\sum_{n=0}^{\infty} (-i)^n \sum_{z_1,\ldots,z_n=y}^{x-1}\mathcal{H}(z_1)\ldots \mathcal{H}(z_n) \Upsilon (\textbf{z}^{(n)})
    \end{equation}
    where $\textbf{z}^{(n)}=z_1z_2\ldots z_n$ and the $n=0$ terms is understood to equal 1.
\end{lemma}

\begin{proof}
Starting with definition \eqref{def Uxy} of $U_{x,y}$, expand the exponentials to obtain, 
\begin{equation}\label{Uxy expanded}
\begin{split}
U_{x,y}&= \bigg( 1-i  \mathcal{H}(x-1)+ (-i)^2\frac{\mathcal{H}(x-1)^2}{2}\cdots \bigg)
    \cdots \bigg( 1-i \mathcal{H}(y)+ (-i)^2\frac{\mathcal{H}(y)^2}{2}+\cdots \bigg).
    \end{split}
\end{equation}
Note that every order $n$ term which one can obtain from expanding the RHS of \eqref{Uxy expanded} has the form,
\begin{equation}\label{generic order n term}
\begin{split}
&(-i)^n\frac{\mathcal{H}(z_1)^{n_1}\cdots \mathcal{H}(z_m)^{n_m}}{n_1!\ldots n_m!},\\
    \end{split}
\end{equation}
with $x-1\geq z_1> z_2 \ldots > z_m\geq y$ and $n_1+\cdots+n_m=n$ for some integer $1\leq m\leq n$. The order $n$ term in $U_{x,y}$ is the sum of all terms of the form \eqref{generic order n term}, \textit{i.e.}, the sum of \eqref{generic order n term} over the ordered partitions of $n$ and over the possible choices of $z_1, z_2 \ldots, z_m$:
\begin{equation}
\begin{split}
&\sum_{m=1}^n \underbrace{\sum_{n_{1}>0}\cdots \sum_{n_{m}>0}}_{n_1+\cdots +n_m=n} \ \sum_{x-1\geq z_1>\ldots >z_n\geq y} (-i)^n\frac{\mathcal{H}(z_1)^{n_1}\cdots \mathcal{H}(z_m)^{n_m}}{n_1!\ldots n_m!}
\\
=&(-i)^n \sum_{z_1,\ldots,z_n=y}^{x-1}H(z_1)\ldots \mathcal{H}(z_n)\Upsilon (\textbf{z}^{(n)}),
    \end{split}
\end{equation}
where in the second line $\Upsilon (\textbf{z}^{(n)})$ enforces the non-strict label ordering $z_1\geq z_2 \ldots \geq z_n$ and provides the factorial factors in the denominator.
\end{proof}

\begin{lemma}\label{lemma BCH appendix} For any operator $\mathcal{O}$ and $x>y\in \mathcal{C}$,
 \begin{equation}\label{BCH appendix equation}\begin{split}
&U_{x,y}^{\dagger}\mathcal{O}U_{x,y} = \sum_{n=0}^{\infty} (-i)^n\sum_{z_1,\ldots, z_n=y}^{x-1} \bigg[\ldots\bigg[\bigg[\mathcal{O},\mathcal{H}(z_1)\bigg],\mathcal{H}(z_2)\bigg]\ldots,\mathcal{H}(z_n)\bigg]\Upsilon (\textbf{z}^{(n)}) ,\\
\end{split}
\end{equation}
where $\textbf{z}^{(n)}=z_1z_2\ldots z_n$ 
 and the $n=0$ term is understood to be $\mathcal{O}$. 
\end{lemma}

\begin{proof}
By lemma \ref{lemma Uxy form} we have,
\begin{equation}\label{eq 1 lemma A2}\begin{split}
&U_{x,y}^{\dagger}\mathcal{O}U_{x,y}=\bigg(\sum_{n_1=0}^{\infty}((-i)^{n_1})^*\hspace{-5mm}\sum_{z_1,\ldots ,z_{n_1}=y}^{x-1}\hspace{-4mm}\mathcal{H}(z_{n_1})\ldots \mathcal{H}(z_1) \Upsilon(\textbf{z}^{(n_1)})\bigg)\hspace{2mm}\mathcal{O}\\&\hspace{60mm}\bigg(\sum_{n_2=0}^{\infty}(-i)^{n_2}\hspace{-5mm}\sum_{z_{n_1+1},\ldots, z_{n_1+n_2}=y}^{x-1}\hspace{-4mm}\mathcal{H}(z_{n_1+1})\ldots \mathcal{H}(z_{n_1+n_2}) \Upsilon(\textbf{z}^{(n_2)})\bigg),
\end{split}
\end{equation}
with $\textbf{z}^{(n_1)}=z_1z_2\ldots z_{n_1}$, $\textbf{z}^{(n_2)}=z_{n_1+1}z_{n_1+2}\ldots z_{n_1+n_2}$ and the $n_1=0$ and $n_2=0$ terms understood to be equal to 1. The order $n>0$ contribution to \eqref{eq 1 lemma A2} is,
\begin{equation}\label{eqtn order n BCH}\begin{split}
     &(-i)^{n}\sum_{z_1,\ldots,z_n=y}^{x-1}\sum_{n_1+n_2=n} (-1)^{n_1} \mathcal{H}(z_{n_1})\ldots \mathcal{H}(z_1) \mathcal{O}   \mathcal{H}(z_{n_1+1})\ldots \mathcal{H}(z_{n})\Upsilon(\textbf{z}^{(n_1)})\Upsilon(\textbf{z}^{(n_2)})\\
     =&(-i)^{n}\sum_{z_1,\ldots,z_n=y}^{x-1}\sum_{n_1+n_2=n} (-1)^{n_1} \mathcal{H}(z_{n_1})\ldots \mathcal{H}(z_1) \mathcal{O}   \mathcal{H}(z_{n_1+1})\ldots \mathcal{H}(z_{n})\sum_{\textbf{w}}\Upsilon(\textbf{w}^{(n)}),
\end{split}
\end{equation}
where in the second line we applied equation \eqref{eq upsilon identity} to write the produce of the $\Upsilon$ as a sum over $\textbf{w}^{(n)}\in\textbf{z}^{(n_1)}\shuffle\textbf{z}^{(n_2)}$. Next, to each $\textbf{w}^{(n)}$ we apply a coordinate transformation so that $\textbf{w}^{(n)}=w_1w_2\ldots w_n\xmapsto{\textbf{w}^{(n)}} \textbf{z}^{(n)}=z_1z_2\ldots z_n$, \textit{i.e.} $w_i\xmapsto{\textbf{w}^{(n)}} z_i$. Writing $p(i)$ to denote the position of $z_i$ in $\textbf{w}^{(n)}$, the transformation we need is $z_i=w_{p(i)}\xmapsto{\textbf{w}^{(n)}} z_{p(i)}$. Plugging this back into \eqref{eqtn order n BCH} we find,
\begin{equation}\label{penultimate step}
\hspace{-5mm} (-i)^{n} \sum_{z_1,\ldots,z_n=y}^{x-1}\Upsilon(\textbf{z}^{(n)})\sum_{n_1+n_2=n}(-1)^{n_1}\sum_{\textbf{w}}
\mathcal{H}(z_{p(n_1)})\ldots \mathcal{H}(z_{p(1)}) \mathcal{O}   \mathcal{H}(z_{p(n_1+1)})\ldots \mathcal{H}(z_{p(n)}),
\end{equation}
where the $p(i)$ labels are implicitly dependent on $\textbf{w}^{(n)}$. By comparing \eqref{penultimate step} with the order $n$ term in \eqref{BCH appendix equation} we find that to complete the proof we must show that,
\begin{equation}\label{counting argument eqtn}\begin{split}
& \bigg[\ldots\bigg[\bigg[\mathcal{O},\mathcal{H}(z_1)\bigg],\mathcal{H}(z_2)\bigg]\ldots,\mathcal{H}(z_n)\bigg]\\
 &=\sum_{n_1+n_2=n}(-1)^{n_1}\sum_{\textbf{w}}
\mathcal{H}(z_{p(n_1)})\ldots \mathcal{H}(z_{p(1)}) \mathcal{O}   \mathcal{H}(z_{p(n_1+1)})\ldots \mathcal{H}(z_{p(n)}).
 \end{split}
\end{equation}
To prove \eqref{counting argument eqtn}, we proceed with a counting argument. Consider the LHS of \eqref{counting argument eqtn} and note that,
\begin{enumerate}
    \item it contains $2^n$ terms;
    \item each term is of the form of $\mathcal{O}$ sandwiched between a product of Hamiltonians;
    \item each of $\mathcal{H}(z_1),\dots \mathcal{H}(z_n)$ appears in each term exactly once;
    \item a term has a positive sign if there is an even number of Hamiltonians to the left of $\mathcal{O}$ and negative otherwise;
    \item the product of Hamiltonians to the right (left) of $\mathcal{O}$ is ordered so that the labels of the $z$ coordinates increase (decrease) from left to right.
\end{enumerate}
We observe that conditions (1-5) are also satisfied by the RHS. In particular, condition (1) follows from the fact that there are $\binom{n}{n_1}$ shuffles of two words of lengths $n_1$ and $n-n_1$ so the number of terms on the RHS is $\sum_{n_1}\binom{n}{n_1}=2^n$, and condition (5) follows from the order-preserving properties of the shuffle. To complete the proof, we argue that there exists exactly $2^n$ terms which satisfy conditions (2-5) and the result follows.

To construct a term satisfying conditions (2-5), fix some $n_1$ and choose $n_1$ Hamiltonians to be on the left of $\mathcal{O}$ leaving the rest to the right. Now order them to satisfy condition (5), and note that there is a unique way to do that. Thus, there are $\binom{n}{n_1}$ ways to construct an appropriate term with $n_1$ Hamiltonians on the left. The total number of possible terms which can be constructed to satisfy conditions (2-5) is therefore given by $\sum_{n_1}\binom{n}{n_1}=2^n$.
\end{proof}

For the remaining lemmas, consider a finite causal set $\mathcal{C}$ and $k$ integers satisfying $|\mathcal{C}|\geq x_1>x_2\ldots >x_k\geq 1$. Let $\mathcal{O}^H(x_i)$ denote a local Heisenberg picture operator living at $x_i$.
 
\begin{lemma}\label{recursion lemma lemma}
    \begin{equation}\label{recursion lemma}
        \mathcal{O}^H(x_1)\ldots \mathcal{O}^H(x_k)=U_{x_k}^{\dagger}\mathcal{O}_{1\ldots k} U_{x_k},
    \end{equation}
    where $\mathcal{O}_{1\ldots k}$ is defined recursively via,
    \begin{equation}\label{def O1...p}
  \mathcal{O}_{1\ldots p} =
    \begin{cases}
      \mathcal{O}(x_{1}) & p=1\\
      U_{x_{p-1},x_p}^{\dagger}\mathcal{O}_{1\ldots p-1}U_{x_{p-1},x_p}\mathcal{O}(x_{p}) & 1<p\leq k.\\
    \end{cases}       
\end{equation}
\end{lemma}

\begin{proof}
This is a short proof by induction. Consider $k=1,2$ for the base case. Then assuming \eqref{recursion lemma} holds for all $k=1,\ldots ,s$ it follows from the composition properties of $U$ (\ref{UU composition}-\ref{UUdagger composition}) that it also holds for $k=s+1$.
\end{proof}

\begin{lemma}
\begin{equation}\label{eq lemma conjugation}
\begin{split}
&U_{x_{p-1},x_p}^{\dagger}\mathcal{O}_{1\ldots p-1}U_{x_{p-1},x_p}\\
&=\sum_{n=0}^{\infty} (-i)^n\sum_{n_1+\cdots +n_{p-1}=n} \ \ \sum_{z_1,\ldots, z_{n_1}=x_2}^{x_1-1}\ \ \sum_{z_{n_1+1},\ldots ,z_{n_1+n_2}=x_3}^{x_2-1}\cdots\sum_{z_{n_1+\cdots n_{p-2}+1},\ldots ,z_{n}=x_p}^{x_{p-1}-1}\\
&\bigg[\hspace{-1mm} \cdots\hspace{-1mm} \bigg[\hspace{-1mm} \cdots\hspace{-1mm} \bigg[\bigg[\hspace{-1mm} \cdots\hspace{-1mm}  \bigg[\mathcal{O}(x_1),\mathcal{H}(z_1)\bigg]\hspace{-1mm} \cdots\hspace{-0.8mm} ,\mathcal{H}(z_{n_1})\bigg]\mathcal{O}(x_2),\mathcal{H}(z_{n_1+1})\bigg]\hspace{-1mm} \cdots\hspace{-0.8mm}  \mathcal{O}(x_{p-1}),\mathcal{H}(z_{n_1+\cdots+ n_{p-2}+1})\bigg]\hspace{-1mm} \cdots\hspace{-0.8mm}  ,\mathcal{H}(z_{n})\bigg]\Upsilon(\textbf{z}^{(n)})\\
\end{split}
\end{equation}
where $\mathcal{O}_{1\ldots p-1}$ is as given in \eqref{def O1...p}.
\end{lemma}

\begin{proof}
    We proceed by induction. Suppose true for $p=s$ and consider the case $p=s+1$. To compute, \begin{equation}\label{short lemma eqtn}
    U_{x_{s},x_{s+1}}^{\dagger}\mathcal{O}_{1\ldots s}U_{x_{s},x_{s+1}},
    \end{equation} we start by expanding $\mathcal{O}_{1\ldots s}$ by applying definition \eqref{def O1...p} together with the inductive assumption,
\begin{equation}\label{eq: Os expanded}\begin{split}
        &\mathcal{O}_{1\ldots s}=\sum_{n=0}^{\infty} (-i)^n\sum_{n_1+\cdots +n_{s-1}=n} \ \ \sum_{z_1,\ldots, z_{n_1}=x_2}^{x_1-1}\ \ \sum_{z_{n_1+1},\ldots ,z_{n_1+n_2}=x_3}^{x_2-1}\cdots\sum_{z_{n_1+\cdots n_{s-2}+1},\ldots ,z_{n}=x_s}^{x_{s-1}-1}\\
&\bigg[\hspace{-1mm} \cdots\hspace{-1mm} \bigg[\hspace{-1mm} \cdots\hspace{-1mm} \bigg[\hspace{-1mm} \cdots\hspace{-1mm}  \bigg[\mathcal{O}(x_1),\mathcal{H}(z_1)\bigg]\hspace{-1mm} \cdots\hspace{-0.8mm} \mathcal{O}(x_2),\mathcal{H}(z_{n_1+1})\bigg]\hspace{-1mm} \cdots\hspace{-0.8mm}  \mathcal{O}(x_{s-1}),\mathcal{H}(z_{n_1+\cdots+ n_{s-2}+1})\bigg]\hspace{-1mm} \cdots\hspace{-0.8mm}  ,\mathcal{H}(z_{n})\bigg]\mathcal{O}(x_s)\Upsilon(\textbf{z}^{(n)}).\\
 \end{split}
\end{equation}
Now, we compute \eqref{short lemma eqtn} by substituting the RHS of \eqref{eq: Os expanded} for $\mathcal{O}_{1\ldots s}$ and applying lemma \ref{lemma BCH appendix} to obtain,
    \begin{equation}\label{eq limits}
        \begin{split}
 &\sum_{m,n=0}^{\infty} (-i)^{m+n}\sum_{n_1+\cdots+n_{s-1}=n} \ \sum_{z_1,\ldots, z_{n_1}=x_2}^{x_1-1}\ \sum_{z_{n_1+1},\ldots, z_{n_1+n_2}=x_3}^{x_2-1}\cdots\sum_{z_{n_1+\cdots n_{s-2}+1},\ldots, z_{n}=x_s}^{x_{s-1}-1}\ \sum_{z_{n+1},\ldots, z_{n+m}=x_{s+1}}^{x_{s}-1}\\
& \bigg[\hspace{-1mm} \cdots\hspace{-1mm} \bigg[\mathcal{O}(x_1),\mathcal{H}(z_1)\bigg]\hspace{-1mm} \cdots\hspace{-0.8mm}  \mathcal{O}(x_{s-1}),\mathcal{H}(z_{n_1+\cdots+ n_{s-2}+1})\bigg]\hspace{-1mm} \cdots\hspace{-0.8mm}  ,\mathcal{H}(z_{n})\bigg]\mathcal{O}(x_s),\mathcal{H}(z_{n+1})\bigg]\cdots \mathcal{H}(z_{n+m})\bigg]\Upsilon(\textbf{z}^{(n)})\Upsilon(\textbf{z}^{(m)}),
        \end{split}
    \end{equation}
with $\textbf{z}^{(m)}=z_nz_{n+1}\ldots z_m$. Now, apply \eqref{eq upsilon identity} to replace the product $\Upsilon(\textbf{z}^{(n)})\Upsilon(\textbf{z}^{(m)})$ with the sum $\sum_{\textbf{z}}\Upsilon(\textbf{z})$ with $\textbf{z}\in \textbf{z}^{(n)}\shuffle\textbf{z}^{(m)}$ and note that the limits of the summations in \eqref{eq limits} mean that the only term which survives is $\Upsilon(\textbf{z}^{(n+m)})$ where $\textbf{z}^{(n+m)}=z_1\ldots z_{n+m}$. Finally to show that the lemma holds for $p=s+1$, relabel as $m+n\rightarrow n$ and $m\rightarrow n_s$.
\end{proof}

\begin{corollary}\label{cor operator product expansion}
    \begin{equation}\label{eq: cor operator product expansion}
    \begin{split}
  &\mathcal{O}^H(x_1)\ldots \mathcal{O}^H(x_k)\\        &=\sum_{n=0}^{\infty} (-i)^n\sum_{n_1+\cdots +n_{k}=n}\ \ \sum_{z_1,\ldots, z_{n_1}=x_2}^{x_1-1}\ \ \sum_{z_{n_1+1},\ldots, z_{n_2}=x_3}^{x_2-1}\cdots \ \ \sum_{z_{n_{1}+\cdots+ n_{k-1}+1},\ldots, z_{n}=1}^{x_{k}-1}\\
&\bigg[\hspace{-1mm} \cdots \bigg[\mathcal{O}(x_1),\mathcal{H}(z_1)\bigg]\hspace{-1mm} \cdots\hspace{-0.8mm} ,\mathcal{H}(z_{n_1})\bigg]\mathcal{O}(x_2),\mathcal{H}(z_{n_1+1})\bigg]\hspace{-1mm} \cdots\hspace{-0.8mm}  \mathcal{O}(x_k),\mathcal{H}(z_{n_1+\cdots+ n_{k-1}+1})\bigg]\hspace{-1mm} \cdots\hspace{-0.8mm}  ,\mathcal{H}(z_{n})\bigg]\Upsilon(\textbf{z}^{(n)}).
    \end{split}
    \end{equation}
\end{corollary}

\begin{proof}
    Follows immediately from plugging \eqref{eq lemma conjugation} into \eqref{recursion lemma}.
\end{proof}

\begin{lemma}\label{lemma correlator form}
    \begin{equation}
    \begin{split}
  &\big\langle \mathcal{O}^H(x_1)\ldots \mathcal{O}^H(x_k)\big\rangle =\sum_{n=0}^{\infty} (-i)^n\sum_{z_1,\ldots, z_{n}=1}^{x_1-1}\bigg\langle \bigg[\hspace{-1mm} \cdots\hspace{-1mm} \bigg[\bigg[\mathcal{O}(x_1)\ldots \mathcal{O}(x_k),\mathcal{H}(z_1)\bigg],\mathcal{H}(z_{2})\bigg]\hspace{-1mm} \cdots\hspace{-0.8mm},  \mathcal{H}(z_{n})\bigg]\bigg\rangle \Upsilon(\textbf{z}^{(n)}).\\
    \end{split}
    \end{equation}
\end{lemma}

\begin{proof}
Begin by taking the expectation value of both sides of \eqref{eq: cor operator product expansion}. Then note that the value of the summand, $$\bigg\langle \bigg[\hspace{-1mm} \cdots\bigg[\mathcal{O}(x_1),\mathcal{H}(z_1)\bigg]\hspace{-1mm} \cdots\hspace{-0.8mm} ,\mathcal{H}(z_{n_1})\bigg]\mathcal{O}(x_2),\mathcal{H}(z_{n_1+1})\bigg]\hspace{-1mm} \cdots\hspace{-0.8mm}  \mathcal{O}(x_k),\mathcal{H}(z_{n_1+\cdots+ n_{k-1}+1})\bigg]\hspace{-1mm} \cdots\hspace{-0.8mm}  ,\mathcal{H}(z_{n})\bigg]\bigg\rangle\Upsilon(\textbf{z}^{(n)}),$$ differs from the value of, $$\bigg\langle \bigg[\hspace{-1mm} \cdots\hspace{-1mm} \bigg[\bigg[\mathcal{O}(x_1)\ldots \mathcal{O}(x_k),\mathcal{H}(z_1)\bigg],\mathcal{H}(z_{2})\bigg]\hspace{-1mm} \cdots\hspace{-0.8mm},  \mathcal{H}(z_{n})\bigg]\bigg\rangle\Upsilon(\textbf{z}^{(n)}),$$ only outside the domain of the nested sum, so replace the former by the latter. (One way to observe this equivalence is via the diagrammatic representations of the expectation value of the nested correlators (cf. Section \ref{sec: in in}). The latter expression is given by all diagrams with $n$ internal and $k$ external vertices such that each internal vertex is connected via a path of retarded propagators to at least one external vertex. The former expression is given by the subset of these diagrams in which only the last $n_{k}$ internal vertices may be connected by a path of retarded propagators to $x_k$, only the last $n_{k}+n_{k-1}$ internal vertices may be connected by a path of retarded propagators to $x_{k-1}$,  \textit{etc.} The diagrams which make up the difference between the two expressions vanish in the domain of the nested sum.)

To complete the proof we note that $\Upsilon(\textbf{z}^{(n)})$ vanishes everywhere outside the domain of the sum which enables us to make the replacement,  \begin{equation}
   \sum_{n_1+\cdots +n_{k}=n}\ \ \sum_{z_1,\ldots, z_{n_1}=x_2}^{x_1-1}\ \ \sum_{z_{n_1+1},\ldots, z_{n_2}=x_3}^{x_2-1}\cdots \ \ \sum_{z_{n_{1}+\cdots+ n_{k-1}+1},\ldots, z_{n}=1}^{x_{k}-1}\longrightarrow \sum_{z_1,\ldots, z_{n}=1}^{x_1-1}.
\end{equation}
\end{proof}

\bibliographystyle{unsrt}
\bibliography{biblio}
\end{document}